\documentclass[11pt]{article}
\usepackage[margin=1in]{geometry}
\usepackage[utf8]{inputenc}
\usepackage[T1]{fontenc}
\usepackage{lmodern}
\usepackage{amssymb,amsmath,amsthm,amsfonts}
\usepackage{textgreek}
\usepackage{mathtools}
\usepackage{enumitem}
\usepackage[numbers,comma,sort&compress]{natbib}
\usepackage{authblk}
\usepackage{graphicx, caption, subcaption}
\usepackage{svg}
\usepackage{float}
\usepackage{qcircuit}
\usepackage{physics}
\usepackage{footnote}
\usepackage{xcolor}
\usepackage{mathrsfs}
\usepackage{bbm}
\usepackage{bm}
\usepackage{hhline}
\usepackage{cases}
\usepackage{url}
\usepackage[colorlinks]{hyperref}
\usepackage{booktabs}
\usepackage{multirow}
\usepackage{graphicx}
\usepackage[capitalise]{cleveref}
\usepackage{pgfplots}
\usetikzlibrary{positioning}
\usepackage[ruled,vlined,linesnumbered]{algorithm2e}
\usepackage{algpseudocode}

\renewcommand{\d}{\mathrm{d}}

\newcommand{\diag}{\operatorname{diag}}

\newcommand{\polylog}{\operatorname{poly}\log}
\newcommand{\poly}{\operatorname{poly}}
\renewcommand{\Tr}{\operatorname{Tr}}

\newcommand{\mc}[1]{\mathcal{#1}}

\newcommand{\wt}[1]{\widetilde{#1}}

\renewcommand{\abs}[1]{\left\lvert#1\right\rvert}
\renewcommand{\norm}[1]{\left\lVert#1\right\rVert}

\newcommand{\ud}{\,\mathrm{d}}
\newcommand{\Or}{\mathcal{O}}

\newcommand{\RR}{\mathbb{R}}
\newcommand{\CC}{\mathbb{C}}
\newcommand{\ZZ}{\mathbb{Z}}
\newcommand{\xx}{\bm{x}}

\newcommand{\zz}{\bm{z}}
\newcommand{\UU}{\bm{U}}
\newcommand{\bsigma}{\bm{\sigma}}
\newcommand{\bphi}{\bm{\phi}}
\newcommand{\bbeta}{\bm{\eta}}

\newtheorem{thm}{\protect\theoremname}
\theoremstyle{plain}

\theoremstyle{plain}
\newtheorem{lem}[thm]{\protect\lemmaname}
\theoremstyle{plain}
\newtheorem{rem}[thm]{\protect\remarkname}
\theoremstyle{plain}
\newtheorem*{lem*}{\protect\lemmaname}
\theoremstyle{plain}
\newtheorem{prop}[thm]{\protect\propositionname}
\theoremstyle{plain}

\newtheorem{defn}[thm]{\protect\definitionname}

\newtheorem{assump}[thm]{\protect\assumptionname}
\theoremstyle{definition}

\providecommand{\definitionname}{Definition}
\providecommand{\assumptionname}{Assumption}
\providecommand{\corollaryname}{Corollary}
\providecommand{\lemmaname}{Lemma}
\providecommand{\propositionname}{Proposition}
\providecommand{\remarkname}{Remark}
\providecommand{\examplename}{Example}
\providecommand{\theoremname}{Theorem}
\providecommand{\conjecturename}{Conjecture}
\providecommand{\assumptionname}{Assumption}

\renewcommand{\ketbra}[2]{\mathinner{|{#1}\rangle\!\langle{#2}|}}

\newcommand{\tildeL}{\widetilde{L}}

\begin{document}

\title{Operator-Level Quantum Acceleration of Non-Logconcave Sampling}
\author[1,2] {Jiaqi Leng}
\author[2,3] {Zhiyan Ding}
\author[2] {Zherui Chen}
\author[2,4,*] {Lin Lin}
\affil[1]{Simons Institute for the Theory of Computing, University of California, Berkeley, USA}
\affil[2]{Department of Mathematics, University of California, Berkeley, USA}
\affil[3]{Department of Mathematics, University of Michigan, Ann Arbor, USA}
\affil[4]{Applied Mathematics and Computational Research Division, Lawrence Berkeley National Laboratory, Berkeley, USA}

\date{}

\maketitle

\renewcommand{\thefootnote}{*}
\footnotetext{\href{mailto:linlin@math.berkeley.edu}{linlin@math.berkeley.edu}}
\renewcommand{\thefootnote}{\arabic{footnote}} 

\begin{abstract} 
  Sampling from probability distributions of the form $\sigma \propto e^{-\beta V}$, where $V$ is a continuous potential, is a fundamental task across physics, chemistry, biology, computer science, and statistics. However, when $V$ is non-convex, the resulting distribution becomes non-logconcave, and classical methods such as Langevin dynamics often exhibit poor performance. We introduce the first quantum algorithm that provably accelerates a broad class of continuous-time sampling dynamics.  For Langevin dynamics, our method encodes the target Gibbs measure into the amplitudes of a quantum state, identified as the kernel of a block matrix derived from a factorization of the Witten Laplacian operator. This connection enables Gibbs sampling via singular value thresholding and yields up to a \textit{quartic} quantum speedup over best-known classical Langevin-based methods in the non-logconcave setting.
  Building on this framework, we further develop the first quantum algorithm that accelerates replica exchange Langevin diffusion, a widely used method for sampling from complex, rugged energy landscapes. 
\end{abstract}

\vspace{4mm}
\textit{\textbf{Significance Statement:}
Sampling from Gibbs distributions under continuous, often non-convex, potentials is a fundamental challenge, which hinders classical methods such as the Langevin dynamics. Existing quantum approaches primarily rely on quantum walks to accelerate classical sampling algorithms, which limits the design space of quantum algorithms to the choice of classical counterparts and inherits technical difficulties in error analysis due to temporal discretization.
We introduce a new, versatile framework for accelerating general sampling processes using quantum computers. In particular, our quantum algorithm is constructed at the operator level without relying on time discretization, thereby providing a new path for quantizing continuous-time processes.}

\section*{Introduction}
Given a continuous potential function $V: \mathbb{R}^d \to \mathbb{R}$ and an inverse temperature $\beta$, sampling from the Gibbs distribution $\sigma \propto e^{-\beta V}$ (sometimes called classical Gibbs sampling) 
is an important task in computational sciences. Originally rooted in
statistical physics, Gibbs sampling has found widespread applications in
computational chemistry~\cite{FrenkelSmit2002},
statistics~\cite{geman1984stochastic,kruschke2010bayesian,RobertCasellaCasella1999}, and learning
theory~\cite{MacKay2003,neal2012bayesian}. 
The Langevin dynamics is a standard continuous-time process for Gibbs sampling.
The overdamped Langevin equation is given by
\begin{equation}\label{eqn:langevin}
\ud X_t=-\nabla V(X_t) \ud t + \sqrt{2/\beta} \ud W_t,
\end{equation}
where $W_t$ is the standard $d$-dimensional Brownian motion. The Gibbs measure $\sigma$ is the stationary distribution of~\cref{eqn:langevin}.
This equation can be approximately solved using a simple scheme known as the Euler-Maruyama discretization, or the unadjusted Langevin algorithm (ULA):
\begin{equation}\label{eqn:em_langevin}
X_{n+1}=-\nabla V(X_n) \Delta t_n + \sqrt{2/\beta} \Delta W_n,
\end{equation}
where $\Delta t_n=t_{n+1}-t_n$ is a time step, and $\Delta W_n\sim \mathcal{N}(0, \Delta t_n I_d)$ follows the Gaussian distribution. While ULA approximates the continuous dynamics~\cref{eqn:langevin}, its discretization can be unstable, which requires a careful choice of time steps to accurately approximate the target stationary distribution. A variant of ULA, the Metropolis-Adjusted Langevin Algorithm (MALA), incorporates a Metropolis–Hastings correction step to ensure the correct equilibrium, yielding a high-accuracy and discrete-time Gibbs sampler~\cite{Bes_95,RR98}. 

The performance of discretized dynamics, such as ULA, MALA, and their variants, has been extensively studied~\cite{DMM19,VW19,DT12,DM17,DM19,Dal17a,Dal17b,Chen_2020,Chewi2020OptimalDD,Lee_2020,Wu_2022,chen2023,Chewi_2024}. Although the exponential convergence of the continuous dynamics~\cref{eqn:langevin} to its stationary distribution $\sigma$ can be established under the sole assumption that $\sigma$ satisfies a Poincar\'e inequality,  analyzing these discrete-time Markov chain Monte Carlo (MCMC) processes is often much more technical and relies on stronger assumptions~\cite{ChewiErdogduLiEtAl2024}.  

In practical applications such as molecular simulations and training generative models, the potential function $V$ is often non-convex, and the target distribution $\sigma$ becomes non-logconcave. Non-logconcave sampling presents a significant challenge for both continuous-time and discrete-time Langevin dynamics, as the convergence rates can become exponentially slow with respect to the barrier height of the potential $V$ and the inverse temperature $\beta$~\cite{RRT17,Che+18,VW19}. This phenomenon is known as metastability~\cite{BOV02,BOV04,BKG05} and is related to the task of rare event sampling~\cite{Bolhuis_2022,Barducci_2011,PhysRevB.66.052301,MARAGLIANO2006168,PMID:18999998,lu2015reactive}.
To overcome such challenges, advanced sampling algorithms such as the replica exchange algorithm (also known as parallel tempering) have been developed, see \cite{Swendsen1986Replica,HukushimaNemoto1996,sugita1999replica,earl2005parallel,doi:10.1021/ct800016r,doi:10.1021/ct100281c,lu2013infinite,lu2019methodological}.  These methods leverage two or a hierarchy of temperature levels to facilitate transitions over barriers that are otherwise prohibitive at low temperatures. 
While non-asymptotic convergence results of replica-exchange-type algorithms exist for certain highly structured models (e.g., Gaussian mixture~\cite{DongTong2022}, mean-field spins~\cite{madras2003swapping}, multimodal distributions~\cite{woodard2009conditions,lee2023improved}, etc.), a general understanding of theoretical efficiency guarantees remains largely open.

Quantum computers offer a promising alternative for accelerating classical sampling processes.
The quantum walk algorithm, first developed by 
Szegedy~\cite{szegedy2004quantum} two decades ago, can be used to achieve a quadratic speedup over classical discrete-time MCMC processes in query complexity. To our knowledge, existing quantum algorithms for classical Gibbs sampling also apply the quantum walk to quantumly encode classical discrete-time MCMC processes, such as  ULA, MALA, and their variants~\cite{childs2022quantum,ozgul2024stochastic,ozgul2025quantumspeedupsmarkovchain}. This restricts the design space of quantum algorithms to classical MCMC frameworks. In addition, the query complexity analysis highly relies on the complexity of the classical discrete-time MCMC process, which inherits the challenges of analyzing the convergence of these discrete-time schemes.

\begin{figure}[ht!]
    \centering
    \includegraphics[height=7cm,trim={{6cm} 0 {6cm} 0},clip]{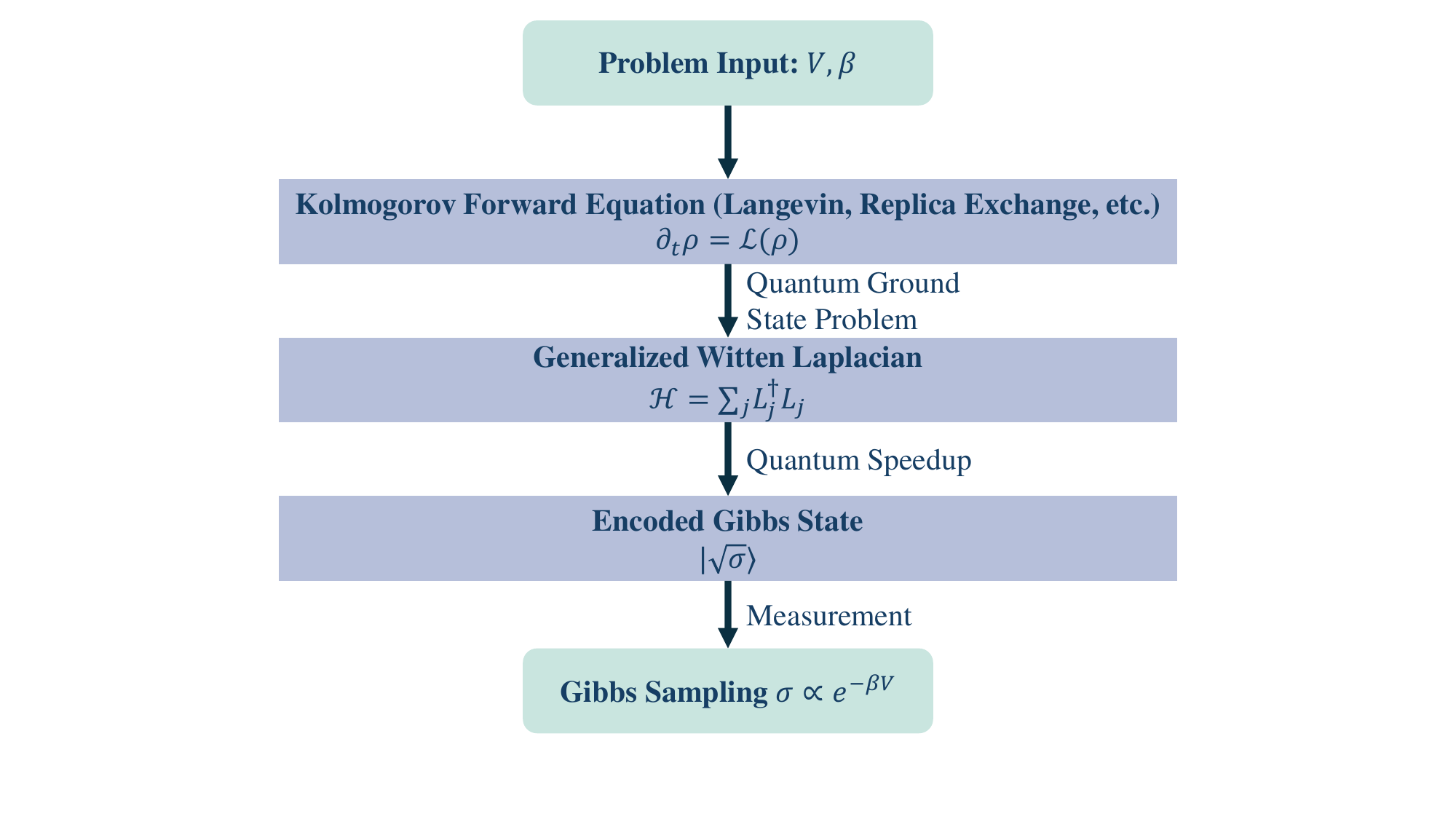}
    \caption{A schematic diagram of operator-level quantum acceleration of classical Gibbs sampling.}
    \label{fig:schematic}
\end{figure}

\vspace{4mm}
In this work, we propose a new framework for achieving quantum acceleration of general sampling processes with the Gibbs measure $\sigma$ as a stationary distribution.
Our quantum algorithm is constructed at the operator level \emph{without} relying on the dynamics generated by a discrete-time classical MCMC process. Take the Langevin dynamics \cref{eqn:langevin} for example. Our method starts from the Fokker--Planck equation (for general Markov processes, it is called the \textit{Kolmogorov forward equation}): 
\begin{equation}\label{eqn:FKPK}
\partial_t \rho = \mathcal{L}(\rho) \coloneqq \nabla \cdot (\rho(t,x)\nabla V(x)) + \beta^{-1} \Delta \rho(t,x).
\end{equation}
Here, $\rho(t,x)$ represents the probability density associated with the random variable $X_t$, and the Gibbs measure is a stationary point of the dynamics, satisfying $\mc{L}(\sigma) = 0$. We refer the reader to~\cref{append:math-prelim} for the relevant mathematical preliminaries.

After a similarity transformation, the generator $\mc{L}$ can be mapped to a quantum Hamiltonian $\mc{H}$, known as the Witten Laplacian operator~\cite{witten1982supersymmetry}, up to a constant scaling factor. The Gibbs measure is then encoded as the ground state of $\mathcal{H}$, denoted by $\ket{\sqrt{\sigma}}$. Measuring this ground state in the computational basis yields $\ket{x}$ with a probability proportional to $\sigma(x)$, thereby achieving Gibbs sampling (see~\cref{fig:schematic}).

In order to efficiently prepare the ground state, a key structural property of $\mathcal{H}$ is that it is \emph{frustration-free}, meaning that it can be decomposed as a sum of operators: $\mathcal{H} = \sum_{i} L_i^{\dag} L_i$, and the \emph{encoded Gibbs state} $\ket{\sqrt{\sigma}}$ is a simultaneous singular vector of all $L_i$'s, satisfying $L_i \ket{\sqrt{\sigma}} = 0$. 
Each $L_i$ is a first-order differential operator free of time-discretization errors, and can be constructed directly from the gradient $\nabla V(x)$.

We will demonstrate that this frustration-free structure provides a simple mechanism to achieve quantum speedup with respect to the continuous-time Langevin dynamics in the non-logconcave setting, whose efficiency depends solely on the Poincar\'e constant (\cref{thm:main}).
Because of the broad applications of continuous sampling problems and the effectiveness of Langevin dynamics, extensive research has examined the complexity of Langevin-based sampling algorithms under various geometric assumptions on the target distribution—such as strong log-convexity~\cite{Wu_2022,childs2022quantum}, log-Sobolev~\cite{osti_10276674,ozgul2024stochastic}, Cheeger~\cite{osti_10276674,ozgul2024stochastic}, and Poincar\'e typed inequalities~\cite{chewi2023log}—in both quantum and classical settings.
\cref{table:comparison} compares the query complexities of our algorithms with existing quantum and classical algorithms.
Although there is a vast body of literature on Langevin-based samplers, our algorithm is the first (classical or quantum) algorithm that provably achieves the $C_{\rm PI}^{1/2}$ scaling for a target distribution with a Poincar\'e constant $C_{\rm PI}$.
In contrast to previous quantum algorithms~\cite{childs2022quantum,ozgul2024stochastic}, our operator-level formulation bypasses time-discretization error analysis and enables a direct continuous analysis via the Witten Laplacian. As a result, we obtain complexity bounds expressed directly in terms of the Poincar\'e constant of the invariant measure, without relying on additional geometric or isoperimetric assumptions such as Cheeger-type inequalities. We do not expect the geometric-constant dependence appearing in~\cite{childs2022quantum,ozgul2024stochastic} to admit substantial improvement, since in the strongly convex regime one has the identity $\gamma^{-1}=C_{\rm CG}^2=C_{\rm PI}$. 
On the other hand, for a broad class of highly non-convex potentials, $C_{\mathrm{PI}}=\wt{\Theta}(C_{\mathrm{CG}})$, where $C_{\rm CG}$ is the Cheeger constant (defined in~\cref{sec:Re}~\cref{eqn:cheeger}), implying that our quantum algorithm can achieve a \textit{quartic} speedup over MALA for non-logconcave sampling tasks.\footnote{Meanwhile, it is worth noting that the $C^2_{\rm CG}$ dependence in classical MALA is highly technical and may be subject to improvement in future work.}
Beyond Langevin-based sampling algorithms, the proximal method achieves a query complexity linear in $C_{\rm PI}$ by doubling the sampling space and employing a careful resampling process~\cite{chewi2023log}. Obtaining a comparable complexity analysis within the framework of Langevin dynamics, to the best of our knowledge, remains an open problem.
A more detailed discussion and comparison can be found in \cref{sec:Re}.

\begin{table}[ht!]
\centering
\scalebox{0.8}{
\begin{tabular}{|c|c|c|c|c|}
\hline
\textbf{Algorithms} & \textbf{Platform} & \textbf{Assumption}& \textbf{Complexity} &\textbf{Warm start}\\
\hline
\textbf{Ours} (\cref{thm:main})&Quantum&\textcolor{green}{$C_{\rm PI}$-Poincar\'e}& $\widetilde{\mathcal{O}}\left(d^{1/2}\textcolor{green}{C^{1/2}_{\rm PI}}\right)$& Yes\\
Quantum MALA v1~\cite[Theorem C.7]{childs2022quantum}&Quantum&$\gamma$-strongly convex& $\widetilde{\mathcal{O}}\left(\textcolor{green}{d^{1/4}}/\gamma^{1/2}\right)$& Yes\\
Quantum MALA v2~\cite[Theorem 5, Remark 27]{ozgul2024stochastic}&Quantum&$C_{\rm CG}$-Cheeger& $\widetilde{\mathcal{O}}\left(d^{1/2}C_{\rm CG}\right)$&  Yes\\
\hline
ULA~\cite[Theorem 6.2.9]{chewi2023log} &Classical& \textcolor{green}{$C_{\rm PI}$-Poincar\'e} & $\widetilde{\mathcal{O}}\left(dC_{\rm PI}^2/\epsilon^{2}\right)$ & No\\
Proximal~\cite[Corollary 8.6.3]{chewi2023log}&Classical&\textcolor{green}{$C_{\rm PI}$-Poincar\'e}& $\widetilde{\mathcal{O}}\left(d^{1/2}C_{\rm PI}\right)$& No\\
MALA~\cite[Theorem 1]{Wu_2022}&Classical
&$\gamma$-strongly convex& $\widetilde{\mathcal{O}}\left(d^{1/2}/\gamma\right)$& Yes \\
MALA~\cite[Lemma 6.5]{osti_10276674}&Classical
&$C_{\rm CG}$-Cheeger& $\widetilde{\mathcal{O}}\left(dC^2_{\rm CG}\right)$& No \\
\hline
\end{tabular}}
\caption{Summary of query complexities of some quantum and classical algorithms for solving the sampling problem when $\beta=1$. The precision parameter $\epsilon$ measures the difference (e.g., TV distance) between the target Gibbs measure and the output random variable. Here, $\widetilde{\Or}$ hides the logarithmic dependence on the precision parameter $\epsilon$. Our warm start assumption is consistent with~\cite[Remark 27]{ozgul2024stochastic} and strictly weaker than that in~\cite[(C.15)]{childs2022quantum}, as discussed in~\cref{rem:warm-start}.
}
\label{table:comparison}
\end{table}

To ensure high success probability, our algorithm requires a good initial state, or a ``warm start'', which has a constant $L^2$-overlap with the target state $\ket{\sqrt{\sigma}}$.
This is a standard assumption in the quantum sampling literature~\cite{childs2022quantum,ozgul2024stochastic}, and strictly weaker than common warm-start assumptions (e.g, $\chi^2$-divergence) in classical literature~\cite{Wu_2022,chewi2023log}.
Several established strategies exist for preparing such states, including quantum simulated annealing~\cite{somma2007quantum,ozgul2024stochastic} and variational quantum circuits. 
In this work, we propose a new approach to prepare this warm-start state based on Lindblad dynamics.
Specifically, we observe that a Lindblad master equation employing the $L_j$ operators (i.e., factors of the Witten Laplacian) as jump operators naturally recovers the Fokker-Planck equation. For potential functions $V$ with moderately simple energy landscapes (for example, those with a constant number of local minima), short-time evolution under this Lindblad dynamics suffices to prepare a warm-start state.

The framework described above can be extended to a broad class of stochastic processes whose infinitesimal generators admit a similar quantization via a similarity transformation, thereby providing a natural and efficient means to accelerate the continuous-time sampling methods without introducing time discretization error. In such cases, the quantum Hamiltonian $\mathcal{H}$ is referred to as the \emph{generalized Witten Laplacian}.
As a concrete demonstration, we apply our method to accelerate the replica exchange Langevin diffusion (\cref{thm:main-2}), which leads to the first provable quantum speedup of continuous-time enhanced sampling dynamics.

\subsection*{Poincar\'e constant}

The generator $\mc{L}$ of the Langevin dynamics \cref{eqn:FKPK} is detailed balanced, namely, the adjoint of the operator $\mathcal{L}$, denoted as $\mathcal{L}^\dagger$, is self-adjoint with respect to an inner product related to the stationary distribution $\sigma$, defined as $\langle f, g \rangle_\sigma \coloneqq \int_{\mathbb{R}^d} \left(f g\right) \sigma\d x$. The smallest positive eigenvalue of $-\mathcal{L}^\dagger$ is called the \emph{spectral gap}, denoted as $\mathrm{Gap}\left(\mathcal{L}^\dagger\right)$. The inverse of the spectral gap, denoted as $C_{\rm PI}$, is the \textit{Poincar\'e constant}. A small spectral gap (or a large Poincar\'e constant) is a strong indicator that the mixing process can be slow. 

\begin{figure}[ht!]
    \centering
    \begin{subfigure}[t]{0.47\linewidth}
        \centering
        \includegraphics[width=\linewidth,trim={1.5cm 0 1.5cm 0},clip]{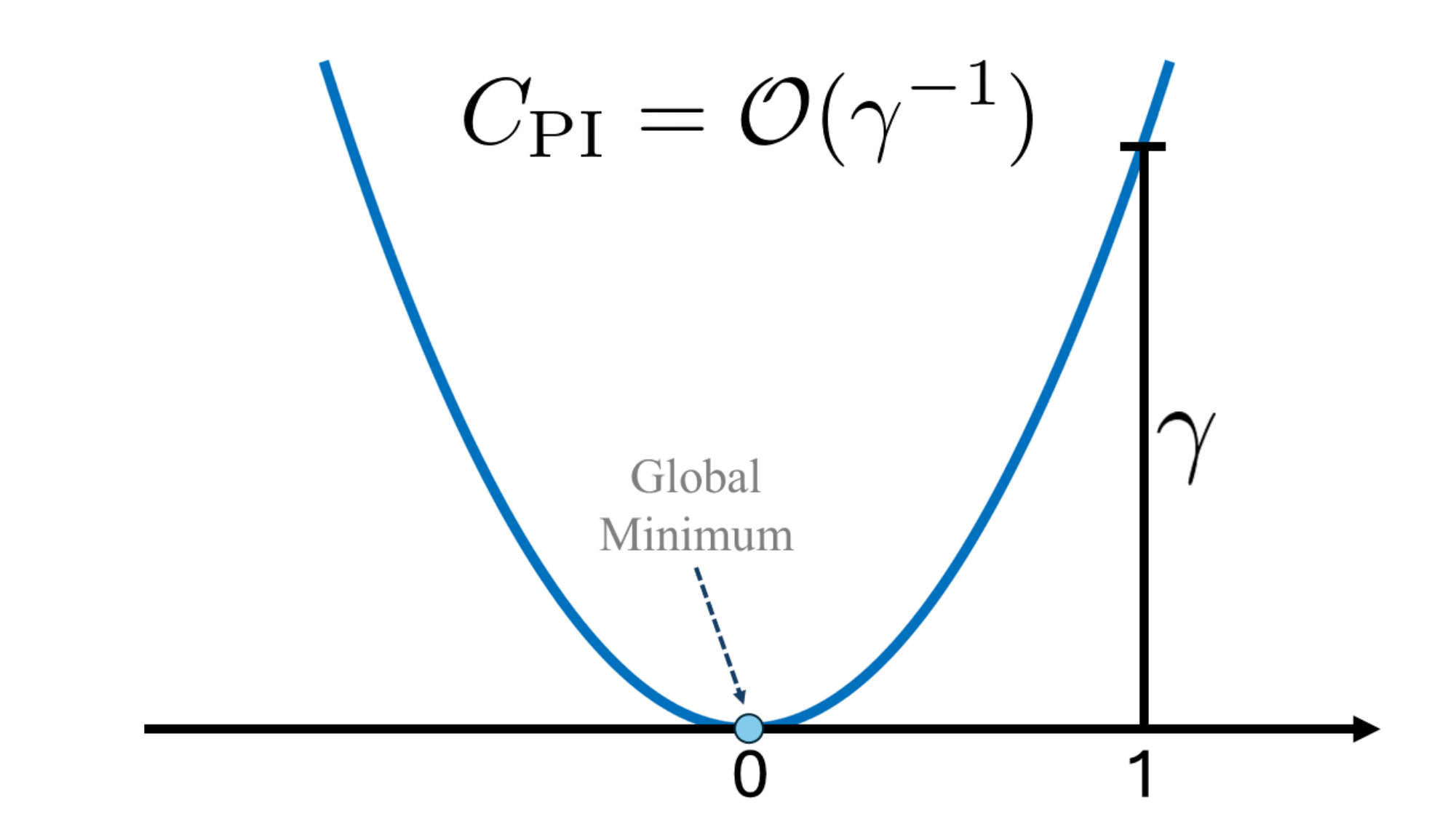}
        \caption{Convex Potential}
        \label{fig:poincare-convex}
    \end{subfigure}
    \hfill
    \begin{subfigure}[t]{0.47\linewidth}
        \centering
        \includegraphics[width=\linewidth,trim={1.5cm 0 1.5cm 0},clip]{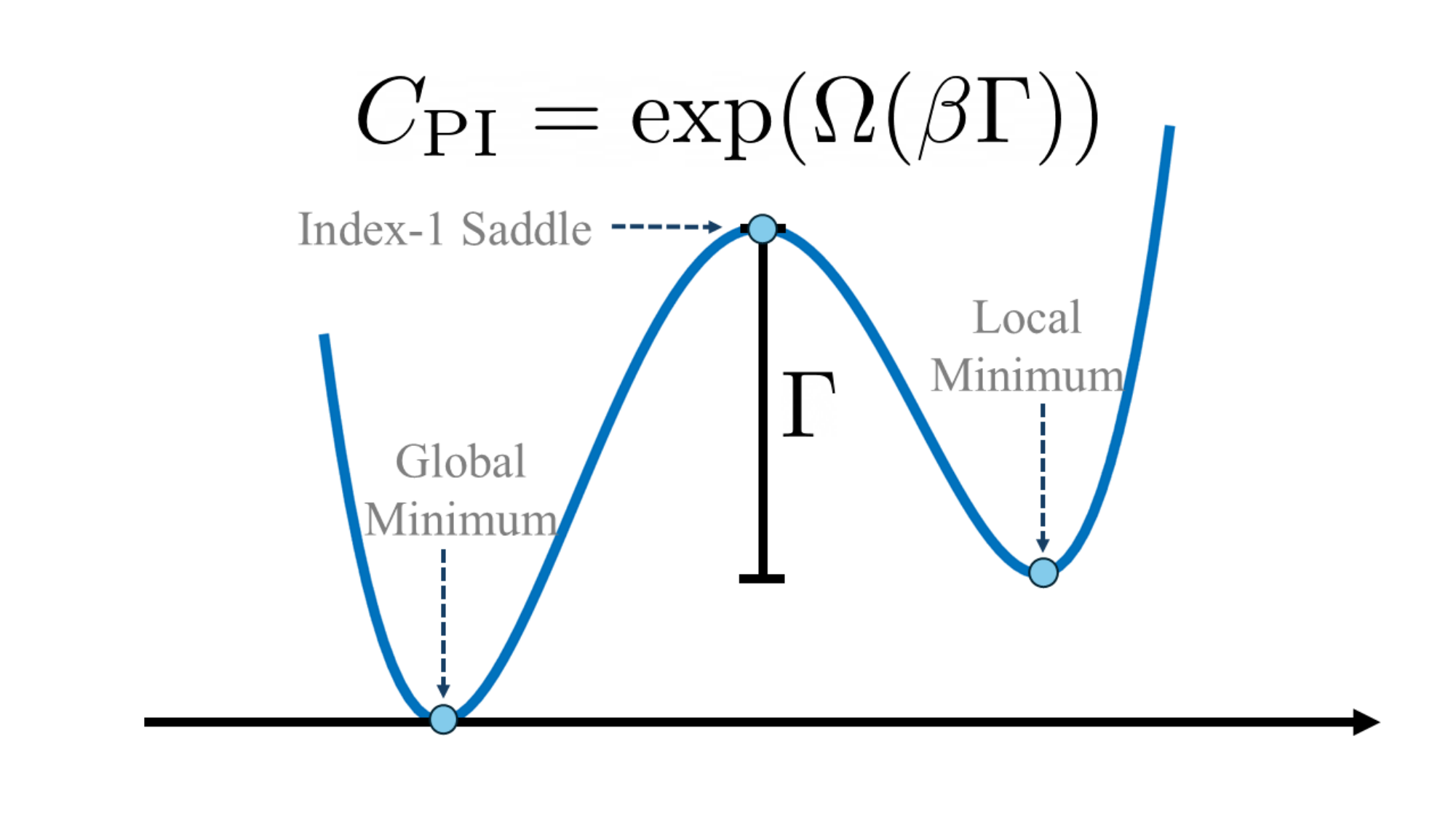}
        \caption{Non-convex Potential}
        \label{fig:poincare-nonconvex}
    \end{subfigure}
    \caption{For Langevin dynamics, the Poincaré constant $C_{\rm PI}$ scales as $\Or(\gamma^{-1})$ for the convex potential, but as $\exp(\Omega(\beta\Gamma))$ for the non-convex potential. The mixing time of the \emph{continuous-time} Langevin dynamics scales linearly in  $C_{\rm PI}$.}
    \label{fig:poincare-figs}
\end{figure}

The spectral gap and the Poincar\'e constant have direct geometric interpretations. If the potential $V(x)$ is $\gamma$-strongly convex (in this case, sampling from $\sigma$ is referred to as \emph{logconcave sampling}), the smallest eigenvalue for the Hessian of $V$ is lower bounded by some constant $\gamma > 0$, i.e., $\nabla^2 V \succeq \gamma I_d$, then  $C_{\rm PI} = \Or(\gamma^{-1})$~\cite[Chapter 9]{villani2009optimal}~\cite[Theorem 4.8.5]{bakry2014analysis}. Thus, the Poincar\'e constant decreases as the potential becomes more convex, causing the Gibbs distribution to concentrate more sharply around the global minimum (\cref{fig:poincare-convex}).

\section*{Operator-Level Acceleration of Langevin Dynamics}
\subsection*{Witten Laplacian}

From the generator $\mc{L}$ in~\cref{eqn:FKPK}, we define a new operator ($\sigma^{\pm 1/2}$ act as multiplicative operators):
\begin{equation}\label{eqn:witten_lap}
    \mc{H} = -\sigma^{-1/2} \circ\mc{L} \circ \sigma^{1/2}=-\frac{1}{\beta} \Delta  + \left(\frac{\beta\|\nabla V\|^2}{4} - \frac{1}{2}\Delta V\right).
\end{equation}
Since $\mc{L}^\dagger$ is self-adjoint with respect to the $\sigma$-weighted inner product, $\mc{H}$ is self-adjoint with respect to the standard $L^2$ inner product.  The operator $\mc{H}$ is called the \emph{Witten Laplacian}\footnote{Some literature defines the Witten Laplacian as  $4\beta^{-1}\mathcal{H}$. Such a constant scaling factor does not affect our analysis. }, which was first proposed by Witten in his analytical proof of the Morse inequalities~\cite{witten1982supersymmetry}. Later, Witten Laplacian became a fundamental tool in the study of metastability~\cite{HelfferNier2004,BKG05}. We may readily check that when $V(x)=\gamma x^2/2$, the Witten Laplacian $\mc{H}$ is the Hamiltonian of a quantum harmonic oscillator with a gap $\gamma$, which is independent of $\beta$.

Due to the similarity transformation, the spectrum of $\mc{H}$ is the same as that of $-\mc{L}$ (and thus of $-\mc{L}^\dagger$), which belongs to the non-negative real axis. 
The kernel of $\mc{H}$ is given by the \emph{encoded Gibbs state} $\ket{\sqrt{\sigma}}$ with $\braket{x}{\sqrt{\sigma}}=\sqrt{\sigma(x)}$, which is a normalized state with respect to the standard $L^2$ inner product. In other words, if the Langevin dynamics is ergodic, then $\ket{\sqrt{\sigma}}$ is the unique ground state of $\mc{H}$ with a spectral gap $\mathrm{Gap}\left(\mc{H}\right)=\mathrm{Gap}\left(\mc{L}^{\dag}\right)$. 
From the encoded Gibbs state we can readily evaluate any classical observable $O(x)$ according to $\bra{\sqrt{\sigma}} O \ket{\sqrt{\sigma}} = \int O \sigma \d x$. Here, we extend the standard $L^2$ inner product to complex numbers and use the bra-ket notation $\langle f | g \rangle:=\langle f, g \rangle=\int \bar{f} g \ud x$, where $\bar{f}$ is the complex conjugation of $f$, and $f,g$ can be interpreted as unnormalized quantum states.

From a PDE perspective, the operators $\mc{L}$ and $\mc{H}$ are related as follows:
\begin{equation}\label{eqn:H_formula}
    \partial_t\rho=\mathcal{L}\rho\ \Leftrightarrow\ \partial_t u=-\mc{H}(u):=\beta^{-1} \Delta u - \left(\frac{\beta\|\nabla V\|^2}{4} - \frac{1}{2}\Delta V\right) u,\quad \rho=\sqrt{\sigma}u\,.
\end{equation}
Since both operators have the same spectral gap, corresponding to~\cref{eqn:chi_square_convergence} in \cref{append:math-prelim}, we have
\begin{equation}\label{eqn:L_square_convergence}
    \chi^2(\rho(t),\sigma)=\left\|u(t)-\sqrt{\sigma}\right\|^2\leq \exp\left(-2 \mathrm{Gap}\left(\mc{H}\right)t\right)\left\|u(0)-\sqrt{\sigma}\right\|^2=\exp\left(-2 \mathrm{Gap}\left(\mc{H}\right)t\right)\chi^2(\rho(0),\sigma).
\end{equation}
Here, we use the fact that $\left\|u(t)-\sqrt{\sigma}\right\|^2$ is equal to the $\chi^2$-divergence. Therefore, $u(t)$ converges exponentially to the encoded Gibbs state in $\chi^2$-divergence.

\subsection*{Quantum algorithm}

The mapping from the generator of Langevin dynamics $\mathcal{L}$ to the Witten Laplacian $\mathcal{H}$ allows us to tackle the Gibbs sampling problem as a ground state preparation problem, which is a standard task in quantum computing~\cite{ge2019faster,lin2020near,papageorgiou2014estimating,augustino2023quantum}.
However, such algorithms inevitably require querying the potential term,  $\frac{\beta\|\nabla V\|^2}{4} - \frac{1}{2}\Delta V$ in~\cref{eqn:H_formula}, which depends on both the first- and second-order derivatives of the potential $V$. This is undesirable since the classical Langevin dynamics only require first-order information. 

The Witten Laplacian admits the following factorization~\cite{witten1982supersymmetry}:
\begin{equation}\label{eqn:H-witten-lap_1}
    \mathcal{H} = \sum^d_{j=1}L^\dagger_j L_j,\quad L_j \coloneqq -i\frac{1}{\sqrt{\beta}} \partial_{x_j} - i\frac{\sqrt{\beta}}{2}\partial_{x_j}V \quad \forall j \in [d]=\{1,2,\dots,d\}.\
\end{equation}
The encoded Gibbs state is simultaneously annihilated by all $L_j$'s, i.e., $L_j \ket{\sqrt{\sigma}}=0$. In other words, $\ket{\sqrt{\sigma}}$ is the ground state of each individual Hamiltonian term $L_i^{\dag} L_i$. Such a Hamiltonian $\mc{H}$ is called \emph{frustration-free}.\footnote{Formally, a Hamiltonian $H = \sum_j H_j$ is frustration-free if the ground state of $H$ is also the ground state for all $H_j$.} 
By concatenating all operators $L_i$ together as 
\begin{equation} 
    \mathbb{L} \coloneqq [L^\top_1, L^\top_2,\dots, L^\top_d]^\top,
    \label{eqn:L-block-matrix}
\end{equation}
\cref{eqn:H-witten-lap_1} can be compactly written as $\mathcal{H} = \mathbb{L}^\dagger \mathbb{L}$. Furthermore, constructing $\mathbb{L}$ only requires the first-order information of the potential $V$. 

Based on this observation, we propose a quantum algorithm for the Gibbs sampling task without solving either the Fokker--Planck equation or the dynamics \cref{eqn:H_formula}. The main idea is to prepare the encoded Gibbs state $\ket{\sqrt{\sigma}}$ as the unique \emph{right singular vector} of $\mathbb{L}$ associated with the zero singular value. The singular value gap of $\mathbb{L}$ is equal to $\mathrm{Gap}\left(\mathbb{L}\right)=\sqrt{\mathrm{Gap}\left(\mc{H}\right)}=1/\sqrt{C_{\rm PI}}$. This provides a source for quantum speedup at the operator level, without either spatial or temporal discretization.

To implement this algorithm on the quantum computer, quantities such as $\ket{\sqrt{\sigma}}$ and $\mathbb{L}$ must be \emph{spatially} discretized (see~\cref{append:pseudo-diff-discretize}). With some abuse of notation, the spatially discretized quantities are still denoted by the same symbols. The spatial discretization error depends on the smoothness property of $V$, which can be systematically controlled (see \cref{append:spatial_discretize_Lj}), and the discussion is omitted here for simplicity. We emphasize that, unlike classical MCMC processes, our algorithm does not require temporal discretization, which often presents significant challenges in the analysis. 

Below, we briefly describe the major steps of our quantum algorithm.

\begin{enumerate}
    \item Prepare an initial state $\ket{\phi}$ satisfying the \emph{warm start condition}, i.e., $\abs{\braket{\phi}{\sqrt{\sigma}}}=\Omega(1)$.
    \item Block-encode the (spatially discretized) matrix $\mathbb{L}$ using a quantum circuit.
    \item Apply a quantum singular value thresholding (SVT) algorithm to $\ket{\phi}$ to filter out contributions in $\ket{\phi}$ corresponding to non-zero singular values of $\mathbb{L}$. This succeeds in preparing a state $\ket{g}\approx \ket{\sqrt{\sigma}}$ with a constant success probability.
    \item Measure the resulting state in the computational basis\footnote{More precisely, the prepared state $\ket{g}$ needs to undergo post-processing to boost the resolution before measurement in the computational basis, see~\cref{lem:high-accuracy-interpolation} for details.} to obtain a sample $x$ that approximately follows the Gibbs distribution $\sigma$.
\end{enumerate}

We now discuss the singular value thresholding algorithm in more detail. The spatially discretized matrix $\mathbb{L}$ is represented by its singular value decomposition (SVD) as $\mathbb{L}=W \Sigma V^{\dag}$, where $V$ is a unitary matrix, $\Sigma$ is a diagonal matrix with non-negative diagonal entries, and $W$ is a rectangular matrix of orthogonal columns. This immediately gives the eigenvalue decomposition of $\mc{H}=V \Sigma^2 V^{\dag}$, i.e., the right singular vectors of $\mathbb{L}$ are the eigenvectors of $\mathcal{H}$, and we identify $\ket{\sqrt{\sigma}}$ with the first column of $V$ corresponding to the unique zero singular value. With the warm start assumption, we can write $\ket{\phi}=Vc$ with $\abs{c_1}=\Omega(1)$. 
Then, we can leverage the Quantum Singular Value Transformation (QSVT) algorithm~\cite{gilyen2019quantum} to implement singular value thresholding. 
Suppose we have access to a block-encoding of $\mathbb{L}$ with a normalization factor $\alpha$. 
We can construct an \emph{even} polynomial $p(s)$ that is approximately equal to $1$ for $0 \le s \le \mathrm{Gap}\left(\mathbb{L}\right)/4\alpha$,  and approximately equal to $0$ for $3\mathrm{Gap}\left(\mathbb{L}\right)/4\alpha \le s \le 1$, see~\cref{fig:rectangle}. 
Applying QSVT with this even polynomial $p(s)$ to the block-encoded $\mathbb{L}$ filters out all singular vectors with singular values greater than $3\mathrm{Gap}\left(\mathbb{L}\right)/4\alpha $, thereby leaving only the encoded Gibbs state. The details of the quantum algorithm, including background on the quantum singular value thresholding algorithm, are provided in~\cref{append:meta-algorithm-qsvt}.

\begin{figure}
    \centering
    \includegraphics[width=0.5\linewidth]{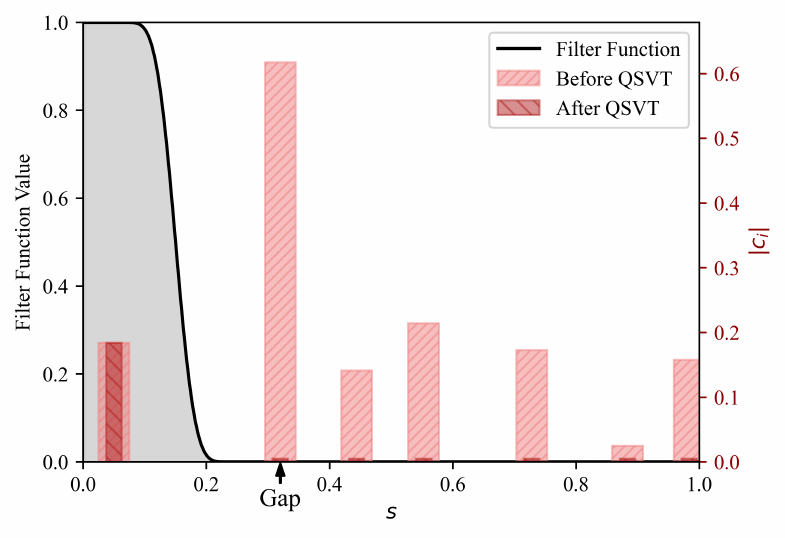}
    \caption{Left y-axis: Value of the filter function $p(s)$ that is approximately equal to $1$ on $[0,\mathrm{Gap}/4]$, and approximately equal to $0$ on $[3\mathrm{Gap}/4, 1]$. Right y-axis: Projection $|c_i|$ to each right singular vector before and after QSVT. Here $\mathrm{Gap}=\mathrm{Gap}\left(\mathbb{L}\right)/\alpha$, and $\alpha$ is the block encoding subnormalization factor for $\mathbb{L}$.}
    \label{fig:rectangle}
\end{figure}

\paragraph{Complexity analysis.}
Our quantum algorithm leverages two main subroutines: constructing the block-encoding of $\mathbb{L}$ and implementing singular value thresholding. 
Each $L_j$ can be represented as a pseudo-differential operator and discretized into a finite-dimensional matrix. We first truncate the space $\RR^d$ into a finite-sized domain $\Omega = [-a,a]^d$ and discretize it using a uniform grid of $N$ points per dimension. Therefore, the dimension of each discretized $L_j$ operator is $N^d$. 
For a smooth potential $V$ with certain growth conditions (e.g., $|V(x)| \ge \gamma \|x\|^2$ for some large $x$ with $\gamma > 0$), to achieve a target sampling accuracy $\epsilon$ in the TV distance, we can choose $a = \mathcal{O}(\log(d/\epsilon))$ and $N = a\cdot \polylog(1/\epsilon)$.
Since the value of the Gibbs measure is very close to zero on the boundary of $\Omega$, we assume a periodic boundary condition to facilitate an efficient block-encoding of $L_j$ using Quantum Fourier Transforms, as detailed in~\cref{append:spatial_discretize_Lj}.
It turns out that the discretized $\mathbb{L}$ can be block-encoded using $2$ quantum queries to the gradient $\nabla V$ and $3$ ancilla qubits, see~\cref{prop:block-encode-A}.
For $\beta \ge 1$, the normalization factor of the block-encoding of $\mathbb{L}$ is equal to $\alpha = \widetilde{O}(\sqrt{\beta d})$, where $\widetilde{\mathcal{O}}(\cdot)$ suppresses poly-logarithmic factors in $d$ and $\epsilon$.
Finally, we employ QSVT to implement singular value thresholding, which can be executed using $\widetilde{O}(\alpha/\mathrm{Gap}\left(\mathbb{L}\right)) = \widetilde{O}(\sqrt{\beta d C_{\rm PI}})$ queries to the block-encoded matrix $\mathbb{L}$. 
The overall complexity of our quantum algorithm is summarized as follows:

\begin{thm}[Informal]\label{thm:main}
Assuming access to a warm start state $\ket{\phi}$, for a sufficiently large $\beta$, there exists a quantum algorithm that outputs a random variable distributed according to $\eta$ such that $\mathrm{TV}(\eta, \sigma) \le \epsilon$ using $\sqrt{\beta d C_{\rm PI}}\cdot \polylog(d,1/\epsilon)$ quantum queries to the gradient $\nabla V$.
\end{thm}      

A rigorous statement and proof of the theorem above are provided in~\cref{thm:main-formal} of~\cref{append:proof-main}. 
To the best of our knowledge, this is the first algorithm that provably achieves a query complexity of $\widetilde{\mathcal{O}}(\sqrt{C_{\rm PI}})$ for general potentials. In the presence of metastability, $C_{\rm PI}$ will grow exponentially with $\beta \Gamma$, where $\Gamma$ denotes the barrier height of the potential $V(x)$. This implies that our quantum algorithm can be used to e.g., simulate systems at a lower temperature.

\section*{Operator-Level Acceleration of Replica Exchange}

In this section, we extend the operator-level quantum algorithms to accelerate the replica exchange Langevin diffusion (RELD).

\subsection*{Replica Exchange Langevin Diffusion}

For rugged (non-convex) energy landscapes, Langevin dynamics may require prohibitively long times to traverse energy barriers separating local minima.
Replica exchange (or parallel tempering) is a widely adopted strategy to mitigate this difficulty.
It involves simulating multiple replicas at different temperatures, where high-temperature replicas explore the global landscape and low-temperature replicas concentrate near local minima.
Through carefully designed swap moves between replicas, replica exchange enables barrier-crossing events that are otherwise rare at low temperatures.
Its ability to accelerate overdamped Langevin diffusion is well established numerically~\cite{sindhikara2010exchange,lu2013infinite,lu2019methodological,chen2020accelerating} and supported to a lesser extent by theoretical results, for instance in the setting of Gaussian mixtures~\cite{dong2022spectral}.

For clarity, we focus on the case of replica exchange with two temperatures. The continuous-time formulation, called the replica exchange Langevin diffusion (RELD), can be modeled by two correlated Langevin dynamics with inverse temperatures $\beta > \beta'>0$, respectively:
\begin{align}\label{eqn:Re_Langevin}
    \d X_t = - \nabla V(X_t) \d t + \sqrt{2/\beta}\d W_t,\quad \d Y_t = - \nabla V(Y_t) \d t + \sqrt{2/\beta'}\d W'_t,
\end{align}
where $X_t$ and $Y_t$ denote the positions of the particle undergoing the low- and high-temperature Langevin dynamics, respectively. $W_t$ and $W'_t$ represent two independent standard Brownian motions.
The invariant distribution of the variable $\{Z_t = (X_t, Y_t)\}_{t\ge 0}$ is a joint Gibbs measure:
\begin{align}
    \bsigma(x,y) \propto \exp\left(-\beta V(x) - \beta' V(y)\right).
\end{align}
The marginal distribution of $x$ under $\bsigma$ recovers the target distribution $\exp(-\beta V(x)) / Z_\beta$.

Without correlation, the mixing time of~\cref{eqn:Re_Langevin} is governed by the low-temperature ($\beta$) system, resulting in no acceleration. In RELD, an additional swap mechanism is introduced to accelerate convergence of~\cref{eqn:Re_Langevin} toward equilibrium. These swap events occur according to a Poisson clock and enable configuration exchange between replicas at different temperatures. Let $\mu > 0$ denote a swapping intensity, and we assume the sequential swapping events take place according to an exponential clock with a rate $\mu$. At a swapping event time, the two particles swap their positions with a Metropolis-Hasting type probability:
\begin{align}\label{eq:mh-swal}
    s(X(t), Y(t)) = \min\left(1, \frac{\bsigma(y,x)}{\bsigma(x,y)}\right).
\end{align}
Note that the Metropolis-Hasting type swapping does not change the invariant Gibbs measure $\bsigma$.

To obtain a continuous-time dynamics that includes the exchange mechanism, let $\rho(t,x,y)$ denote the probability density of the random variables $(X_t, Y_t)$ for time $t$.
The evolution of $\rho(t,x,y)$ can be characterized by the forward Kolmogorov equation $\partial_t \rho = \mathcal{L}(\rho)$, where the generator $\mathcal{L}$ captures both the Langevin dynamics and the swapping mechanism. In particular, $\mathcal{L}$ takes the following form:
\begin{equation}
    \mathcal{L} = \mathcal{L}_1 + \mathcal{L}_2 + \mathcal{L}_s,
\end{equation}
where $\mathcal{L}_1$ and $\mathcal{L}_2$ are the Fokker--Planck generators in~\cref{eqn:FKPK} with inverse temperatures $\beta$ and $\beta'$, respectively. $\mathcal{L}_s$ corresponds to the Metropolis-Hasting swapping described by~\cref{eq:mh-swal} with intensity $\mu$ as
\begin{equation}
    \mathcal{L}_s(\rho(x,y)) \coloneqq \mu \left[s(y,x)\rho(y,x) - s(x,y)\rho(x,y)\right].
\end{equation}
Under mild conditions, the dynamics can be shown to be ergodic, and the joint Gibbs measure $\bsigma$ is the unique invariant measure of the generator $\mathcal{L}$~\cite{DupuisLiuPlattnerEtAl2012}.

For non-logconcave sampling, low-temperature Langevin dynamics often become trapped near metastable configurations for an exponentially long time. Through the swapping process, well-explored configurations in the high-temperature system can be transferred to the low-temperature system, thereby significantly reducing the overall mixing time (see \cite{DongTong2022} for an analysis of RELD for Gaussian mixture models). 

While the swapping mechanism enables replica exchange to effectively overcome energy barriers in the potential landscape and thus achieve faster mixing, it could still suffer from the curse of dimensionality. For example, \cite{woodard2009conditions,lee2023improved} suggest that in the presence of multiple modes, the spectral gap of replica exchange can decay rapidly as the problem dimension increases. Thus it is desirable to design  a quantum algorithm with provable speedup over RELD.

\subsection*{Generalized Witten Laplacian of RELD}

Now, we consider the operator transformation
\begin{equation}\label{eq:reld-witten-laplacian}
    \mathcal{H} \coloneqq - \bsigma^{-1/2}\circ \mathcal{L} \circ \bsigma^{1/2},
\end{equation}
which is referred to as the \textit{generalized Witten Laplacian} of RELD. 
Due to the similarity transformation, the spectrum of $\mathcal{H}$ is the same as that of $-\mathcal{L}$, and the kernel of $\mathcal{H}$ is given by the extended \textit{encoded Gibbs state} $\ket{\sqrt{\bsigma}}$ with $\braket{x,y}{\sqrt{\bsigma}} = \sqrt{\bsigma(x,y)}$. Note that $\ket{\sqrt{\bsigma}}$ is a product state, and the Gibbs distribution $\sigma \propto e^{-\beta V}$ can be recovered by measuring $\ket{\sqrt{\bsigma}}$ in the $x$-variable register.
Under this transformation, the two Fokker--Planck-type operators $\mathcal{L}_1$ and $\mathcal{L}_2$ are mapped to two Witten Laplacians, denoted by $\mathcal{H}_1$ and $\mathcal{H}_2$. Similar to~\cref{eqn:H-witten-lap_1}, each Witten Laplacian can be factorized as the sum of $d$ non-negative operators of the form $L^\dagger L$.

The transformation of the swap operator $\mathcal{L}_s$ is more involved. First, we rewrite the operator $\mathcal{L}_s = \mu (W - I)\circ S$, where $I$ is the identity operator, $W(\psi(x,y)) \coloneqq \psi(y,x)$ interchanges the $x$ and $y$ variables in a function, and $S(\psi(x,y)) \coloneqq s(x,y)\psi(x,y)$ represents the multiplication with the function $s$.
It can be readily verified that $\sqrt{s(x,y)\bsigma(x,y)} = \sqrt{s(y,x)\bsigma(y,x)}$. In other words, $W$ commutes with $S^{1/2}\bsigma^{1/2}$, i.e., $W S^{1/2}\bsigma^{1/2} = S^{1/2}\bsigma^{1/2} W$.
Therefore, we can rewrite the transformed swap operator
\begin{align}\label{eqn:L_s_similarity}
    \mathcal{H}_s \coloneqq - \bsigma^{-1/2}\circ \mathcal{L}_s \circ \bsigma^{1/2} = \mu S^{1/2}(I-W) S^{1/2} = L^\dagger_s L_s,\quad L_s = \sqrt{\mu/2}(I-W) S^{1/2}.
\end{align}
Note that we use $W^2=I$ in the last step.
It turns out that the generalized Witten Laplacian of RELD consists of three component operators $\mathcal{H} = \mathcal{H}_1+\mathcal{H}_2+\mathcal{H}_s$. Moreover, $\mathcal{H}$ admits a factorization of the form 
\begin{equation}\label{eqn:L_RE}
    \mathcal{H} = \mathbb{L}_{\mathrm{RE}}^\dagger \mathbb{L}_{\mathrm{RE}},\quad \mathbb{L}_{\mathrm{RE}} = [L^\top_1, \dots, L^\top_d, (L'_1)^\top,\dots, (L'_d)^\top, L^\top_s]^\top,
\end{equation}
where the block matrix $\mathbb{L}_{\mathrm{RE}}$ encompasses $(2d+1)$ operators: $\{L_j\}^d_{j=1}$ and $\{L'_j\}^d_{j=1}$ correspond to the Langevin dynamics with inverse temperature $\beta$ and $\beta'$, respectively; $L_s$ corresponds to the swap mechanism. More details of $\mathbb{L}_{\mathrm{RE}}$ can be found in~\cref{append:symmetrize-reld}.
We refer to $\mathcal{H}$ as the \textit{generalized Witten Laplacian} of RELD, and it enables us to apply our operator-level quantum sampling algorithm.

\subsection*{Quantum algorithms and complexity analysis}

The encoded (joint) Gibbs state $\ket{\sqrt{\bsigma}}$ can be prepared as the ground state of the generalized Witten Laplacian $\mathcal{H}$~\cref{eqn:L_RE}, or equivalently, the right singular vector of $\mathbb{L}_{\rm RE}$ associated with the zero singular value. This can be realized via the singular value thresholding algorithm as discussed before. 

To construct a block-encoding of $\mathbb{L}_{\rm RE}$, we need to perform spatial discretization for the component operators $L_j$ in $\mathbb{L}_{\rm RE}$ for each $j \in [2d+1]$. Note that the component operator $L_j$ in $\mathbb{L}_{\rm RE}$ acts on functions in $\RR^{2d}$. Similar to the previous section, we can first truncate the space $\RR^{2d}$ into a finite-size numerical domain $\Omega = [-a,a]^{2d}$ and discretize it using a uniform grid with $N^{2d}$ points. The discretized $L_j$ operator is a matrix of dimension $N^{2d}$. For a smooth potential with certain growth conditions, to achieve a target sampling accuracy $\epsilon$ in the TV distance, we can choose $a = \mathcal{O}(\log(d/\epsilon))$ and $N = a\cdot\polylog(1/\epsilon)$.

Since the first $2d$ component operators in $\mathbb{L}_{\rm RE}$ are essentially the same as in the Langevin dynamics case, they can be block-encoded using $2$ queries to $\nabla V$ through the same technique as mentioned before.
The last operator $L_s \propto (I-W)\circ S^{1/2}$ can be implemented by concatenating two quantum circuits, one for $I-W$ and another for the scaler multiplication $S^{1/2}$. The exchange operator $W$ can be realized by $d\lceil\log_2(N)\rceil$ SWAP gates; the scalar multiplication operator $S^{1/2}$ can be directly implemented using an arithmetic circuit and queries to the function $V$.
Overall, we can build a block-encoding of $\mathbb{L}_{\mathrm{RE}}$ with $2$ quantum queries to the gradient $\nabla V$, $4$ quantum queries to the function value of $V$, and $10$ ancilla qubits, as detailed in~\cref{prop:block-encoding-L-RE}. 
The normalization factor of this block-encoding is equal to $\alpha \le \widetilde{\mathcal{O}}(\sqrt{\beta d})$ (suppressing poly-logarithmic factors in $d$ and $\epsilon$).
Finally, we can implement QSVT for singular value thresholding, which creates a projection onto the encoded Gibbs state $\ket{\sqrt{\bsigma}}$ using $\widetilde{\mathcal{O}}\left(\sqrt{\beta d/\mathrm{Gap}\left(\mathcal{L}^\dagger\right)}\right)$ queries to the block-encoded matrix $\mathbb{L}_{\rm RE}$. Here, $\mathrm{Gap}\left(\mathcal{L}^\dagger\right)$ denotes the spectral gap of the generator of RELD, i.e., the smallest positive eigenvalue of $-\mc{L}^\dagger$. The overall complexity of the quantum algorithm is summarized below:

\begin{thm}[Informal]\label{thm:main-2}
    Assuming access to a warm start state $\ket{\bphi}$, and let $\beta'$, $\mu$ be constant independent of $d$ and $\beta$. 
    For a sufficiently large $\beta$, there exists a quantum algorithm that outputs a random variable $(X,Y)\in \RR^{2d}$ distributed according to $\bbeta$ such that $TV(\bbeta, \bsigma) \le \epsilon$
    using $\sqrt{\beta d/\mathrm{Gap}\left(\mathcal{L}^\dagger\right)} \cdot \polylog(d, 1/\epsilon)$ quantum queries to the function value $V$ and the gradient $\nabla V$, respectively. Consequently, the distribution of $X$ with a marginal law $\eta$ satisfies $TV(\eta, \sigma) \le \epsilon$, where $\sigma \propto e^{-\beta V}$.
\end{thm}

A rigorous statement and proof of the theorem above are provided in~\cref{thm:main-2-formal} of~\cref{append:proof-main-2}.
The complexity in the above theorem scales as $\mathcal{O}\left(\sqrt{1/\mathrm{Gap}}\right)$, where $\mathrm{Gap}$ denotes the spectral gap of the continuous-time RELD.

\section*{Lindblad Dynamics for Warm-Start Preparation}
We consider a Lindblad master equation given by:
\begin{equation}\label{eqn:lindbladian}
    \partial_t \rho = \mathfrak{L}[\rho], \quad \mathfrak{L}[\rho] \coloneqq  \sum^d_{j=1}\left(2 L_j\rho L^\dagger_j - \left\{L^\dagger_j L_j, \rho\right\}\right),
\end{equation}
where $\rho(t)$ is a time-dependent density operator and the jump operators $\{L_j\}^d_{j=1}$ are the same as in~\cref{eqn:H-witten-lap_1}. Here, $\{A,B\} = AB + BA$ represents the anti-commutator.
It is straightforward to verify that the encoded Gibbs state is a fixed point of the Lindbladian $\mathfrak{L}$, since $L_j \ket{\sqrt{\sigma}} = 0$ for all $j$. In general, $\rho(t)$ is a mixed state. However,  the evolution of its ``diagonal elements'' in the computational basis reduces to the Fokker--Planck equation in~\cref{eqn:FKPK}; see \cref{lem:weak-convergence}.

Although the Lindblad dynamics is not asymptotically faster in convergence, it can rapidly prepare a quantum representation of a metastable state.  This behavior arises because, while the spectral gaps between metastable states are typically small, which leads to a large Poincar\'e constant and hence slow overall convergence, the metastable subspace itself is often well separated from the rest of the Lindbladian spectrum. This can happen for potential functions $V$ with a relatively small number of local minima and for temperatures that are not too low.
As a result, the system quickly relaxes into this low-lying subspace, yielding an intermediate state with a sufficiently large overlap with the target Gibbs state, which serves as an effective warm start for the subsequent quantum acceleration stage.

Therefore, by simulating the Lindblad equation on a quantum computer, which uses the same block-encoding of $\mathbb{L}$~\cite{li2022simulating}, we can efficiently obtain a warm-start state for our SVT-based quantum algorithm.
A numerical demonstration of this two-phase behavior is presented in the next section, and further technical details are provided in Appendix~\ref{append:lindblad-warm-start}.

\section*{Numerical Experiments}
In this section, we numerically demonstrate our quantum algorithms to accelerate Langevin dynamics (LD) and replica exchange Langevin diffusion (RELD) for non-logconcave sampling. Our results show that the quantum algorithms exhibit significant speedups over their classical counterparts. To simulate our quantum algorithms, we numerically implement singular value thresholding. The details of the numerical implementation can be found in Appendix \ref{append:detail_num}.

\subsection*{Quantum-accelerated Langevin dynamics}

\begin{figure}[ht!]
    \centering
    \includegraphics[width=0.95\linewidth]{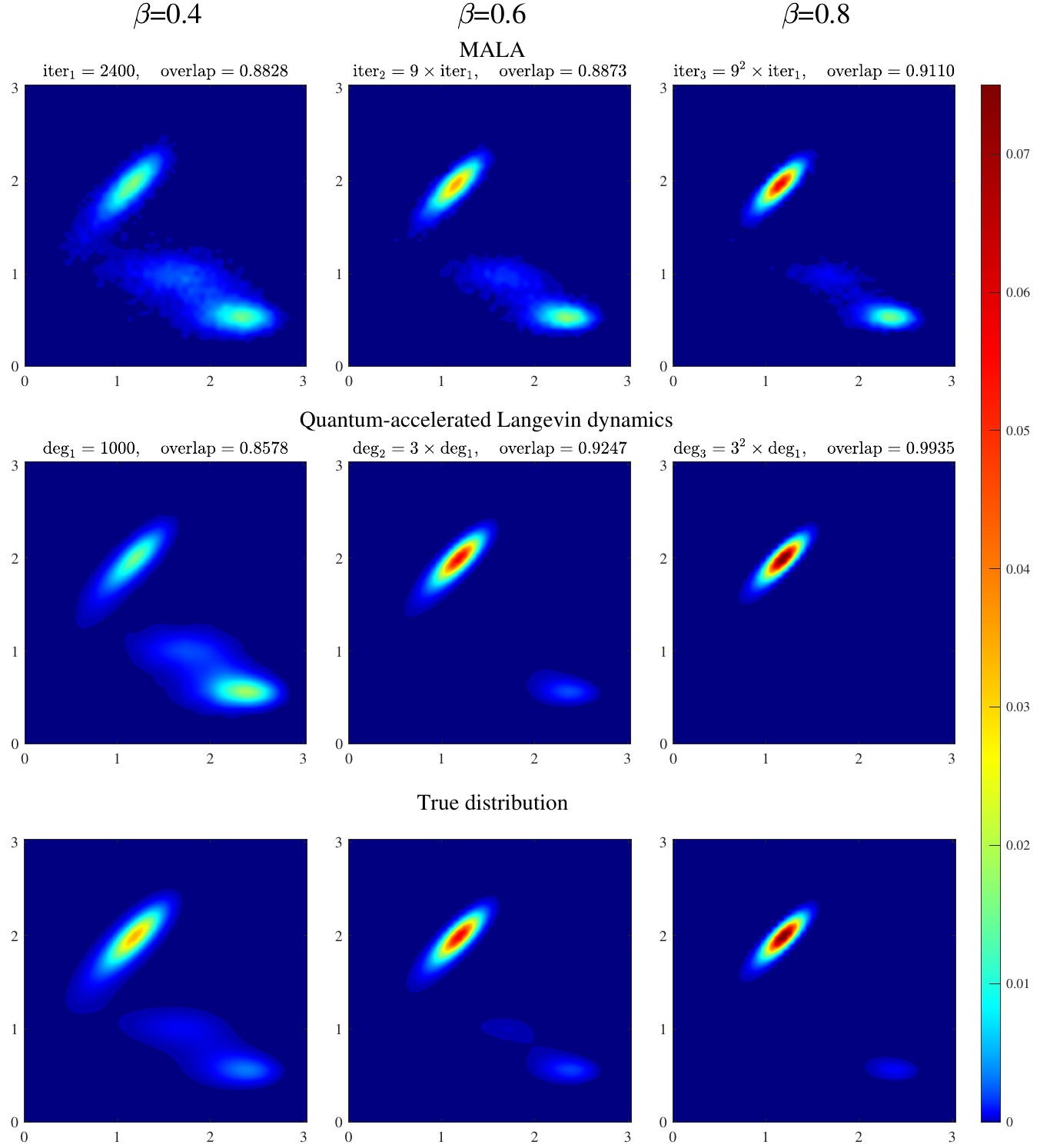}
    \caption{Quantum acceleration of Langevin dynamics. Each column corresponds to a different inverse temperature $\beta$.
    The overlap $\abs{\bra{\sqrt{\sigma}} \ket{\phi}}$ measures the similarity between the sampled and true distribution, where $\ket{\sqrt{\sigma}}$ is the encoded Gibbs state and $\ket{\phi}$ represents either (1) the square root of the output probability density by MALA, or (2) the output pure state by our method. Note that the overlap definition can be generalized to accommodate mixed states; see~\cref{eqn:mixed-state-overlap}.
    Top row: Sample distributions obtained using MALA with the number of iterations increasing by a factor of 9. Middle row: Distributions obtained by our quantum algorithm with the degree of polynomials increasing by a factor of 3. Bottom row: True distribution (Gibbs state).
    }
    \label{fig:qsvt-muller-brown}
\end{figure}

We compare the performance of our quantum-accelerated Langevin sampling with that of MALA using the \textit{M\"uller-Brown} potential~\cite{MuellerBrown1979}, which is often used as a proof of concept for rare event sampling in computational chemistry.
The M\"uller-Brown potential is characterized by a highly non-convex energy surface with three local minima. Details of this potential function, including its analytical expression and energy landscape, are provided in~\cref{append:detail_QALD}.

To conduct a fair comparison, both algorithms start from the same Gaussian distribution $\rho_0$ (or the corresponding Gaussian state $\ket{\sqrt{\rho_0}}$) centered at one of the local minima of the potential function.
We choose three inverse temperatures: $\beta = 0.4, 0.6, 0.8$ for the experiment.
For all tested inverse temperatures, the overlap between the initial distribution $\rho_0$ and the target Gibbs state $\sigma$, i.e., $\abs{\braket{\sqrt{\rho_0}}{\sqrt{\sigma}}}$, is approximately 0.1, indicating a reasonable warm start.

To reflect the computational complexity of the algorithms, we report the number of MALA iterations and the polynomial degree of the filter function in QSVT, both of which directly determine the number of queries to $\nabla V$.
Due to the non-convexity of the potential, the Arrhenius law predicts that the mixing time of MALA scales exponentially with the inverse temperature. For an ascending sequence $\beta \in \{0.4, 0.6, 0.8\}$, the numerical results show that MALA requires approximately 9 times more iterations for every $0.2$ increase in $\beta$, as illustrated in the top row of~\cref{fig:qsvt-muller-brown}.
On the other hand, the middle row of~\cref{fig:qsvt-muller-brown} shows that our quantum algorithm outputs a sample distribution with a comparable (or even slightly better) accuracy by increasing the QSVT polynomial degrees by a factor of $3$ for the same increment in $\beta$.
This observation highlights the efficiency of our quantum algorithm and demonstrates the desired speedup in query complexity over classical MALA.

\subsection*{Quantum-accelerated replica exchange}
We numerically demonstrate the quantum-accelerated RELD using a 1-dimensional potential:
\begin{align}
    V(x) = \cos(\pi x)^2 + 0.25x^4,
    \label{eq:1D_hard_2850}
\end{align}
which is a non-convex function with $4$ local minima, as illustrated in~\cref{fig:1D-hard}.
Even though this experiment is conducted on a one-dimensional example, we expect that similar behavior can be observed in a high dimensional setting, when the energy barrier can be identified along a ``reaction coordinate'' as is often the case in computational chemistry~\cite{ERenVanden-Eijnden2005,ChipotPohorille2007}.

\begin{figure}[ht!]
    \centering
    \begin{subfigure}[t]{0.38\linewidth}
        \centering
        \includegraphics[trim={6cm 0 6cm 0},clip, width=1\linewidth]{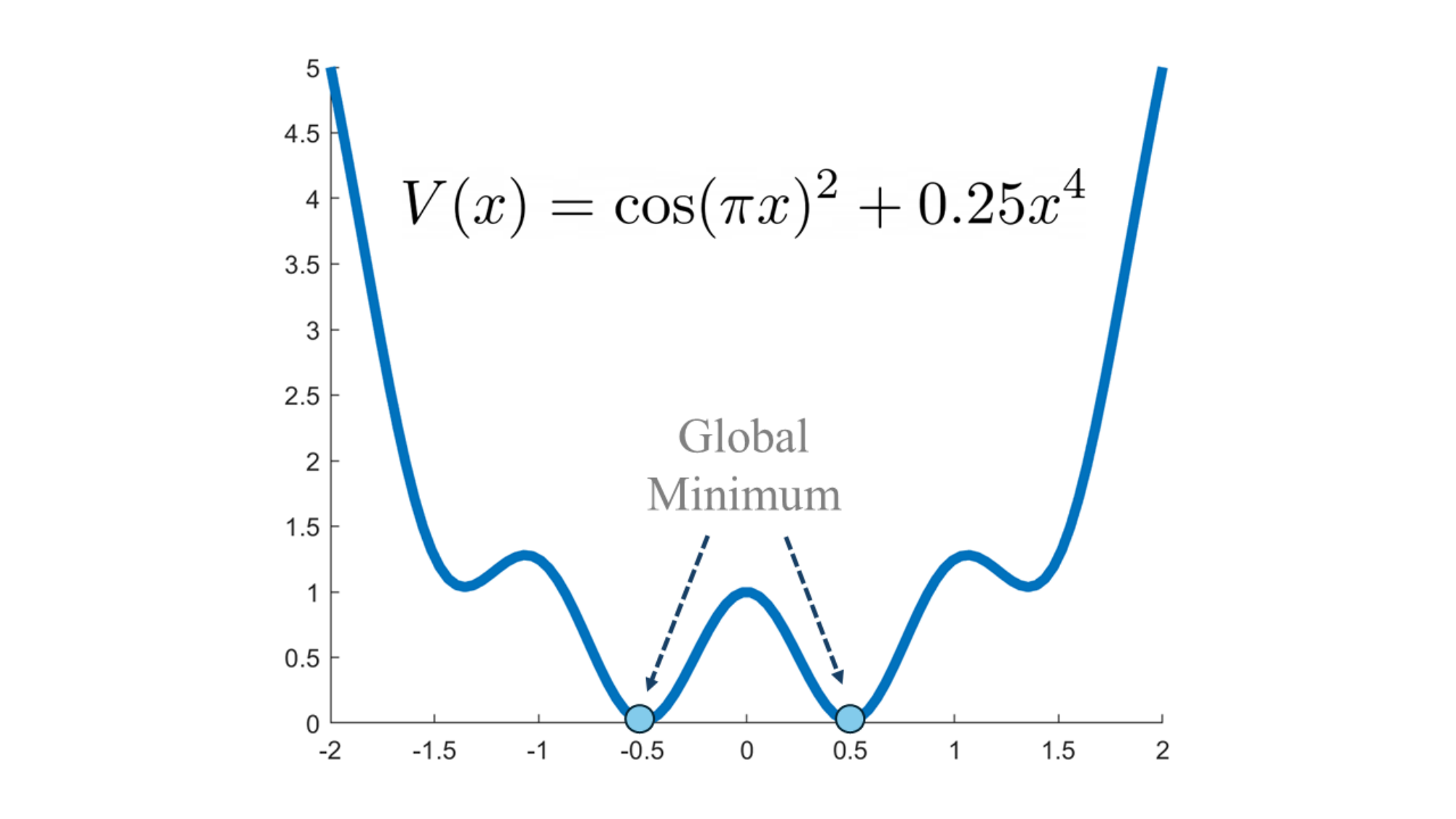}
        \caption{Energy Surface of $V(x)$}
        \label{fig:1D-hard}
    \end{subfigure}
    \hfill
    \begin{subfigure}[t]{0.54\linewidth}
        \centering
        \includegraphics[width=\linewidth]{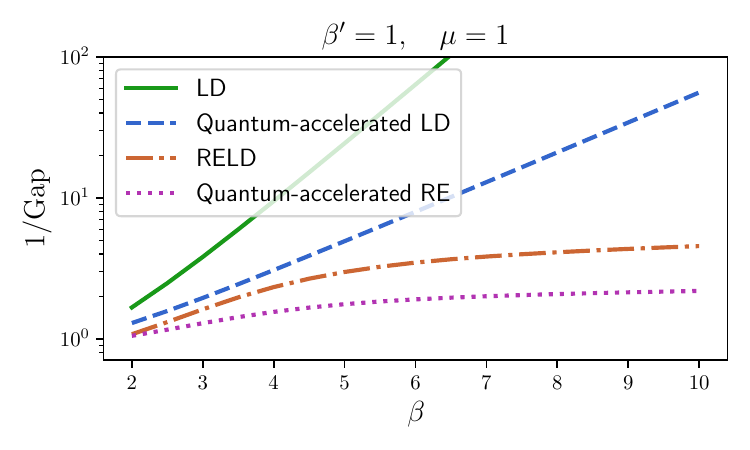}
        \caption{Inverse spectral gaps of different dynamics}
        \label{fig:re-spec-gap}
    \end{subfigure}
    \caption{Quantum acceleration of replica exchange. Left: a non-convex 1D potential. Right; Inverse spectral gaps for LD, RELD, and their quantum counterparts as functions of $\beta$. The spectral gap of LD decays exponentially in $\beta$, while that of RELD shows a much milder decay. The quantum algorithms achieve square-root improvements in the gaps in both cases.}
    \label{fig:combined-fig}
\end{figure}

In~\cref{fig:re-spec-gap}, we illustrate the inverse spectral gaps of two classical dynamics (LD and RELD) and their quantum counterparts. The inverse spectral gap indicates the scaling of the query complexity (for both quantum and classical algorithms) as $\beta$ grows. The spectral gaps of LD and RELD are computed based on the Witten Laplacians, i.e., \cref{eqn:witten_lap} and \cref{eq:reld-witten-laplacian}, respectively. For the two corresponding quantum algorithms, the singular value gaps are computed using~\cref{eqn:L-block-matrix} and~\cref{eqn:L_RE}, respectively. For RELD, we set $\beta' = \mu = 1$ for all choices of $\beta$.

We observe that the spectral gap of Langevin dynamics (LD) decays exponentially as $\beta$ increases, while RELD exhibits a much milder decay, illustrating the advantage of RELD for non-logconcave sampling. Furthermore, we find that the singular value gaps of the quantum-accelerated algorithms are approximately\footnote{Due to spatial discretization, the singular value gap in the quantum algorithms can be slightly different from the square root of the Witten Laplacian gap. This error can be exponentially small thanks to the pseudo-spectral method.} the square root of the spectral gaps of their classical counterparts, consistent with our theoretical analysis and confirming that our quantum algorithms achieve a quadratic speedup as $\beta$ becomes large.

In~\cref{fig:combined-fig-lindblad-warm-start}, we demonstrate the effectiveness of the Lindbladian-based warm-start preparation method using the same 1D potential as in~\cref{fig:1D-hard}.
We evolve the Lindblad dynamics starting from an initial state with negligible overlap with the encoded Gibbs state (see~\cref{fig:warminit}).
The overlap between the solution state $\rho(t)$ and the encoded Gibbs state is defined as 
\begin{equation}\label{eqn:mixed-state-overlap}
    \sqrt{\Tr\left[\ketbra{\sqrt{\sigma}}{\sqrt{\sigma}}\rho(t)\right]},
\end{equation}
and~\cref{fig:warmVsMix} illustrates how the overlap changes over time.
We observe that a short-time Lindbladian evolution ($t=1$) already yields a warm-start state with a constant ($\ge 0.2$) overlap.
Beyond this point, the dynamics requires an evolution time at least an order of magnitude longer to fully converge to the target Gibbs state. This later stage can be significantly accelerated using our quantum-accelerated LD or RELD algorithms.
\begin{figure}[ht!]
    \centering
    \begin{subfigure}[t]{0.45\linewidth}
        \centering
        \includegraphics[width=\linewidth]{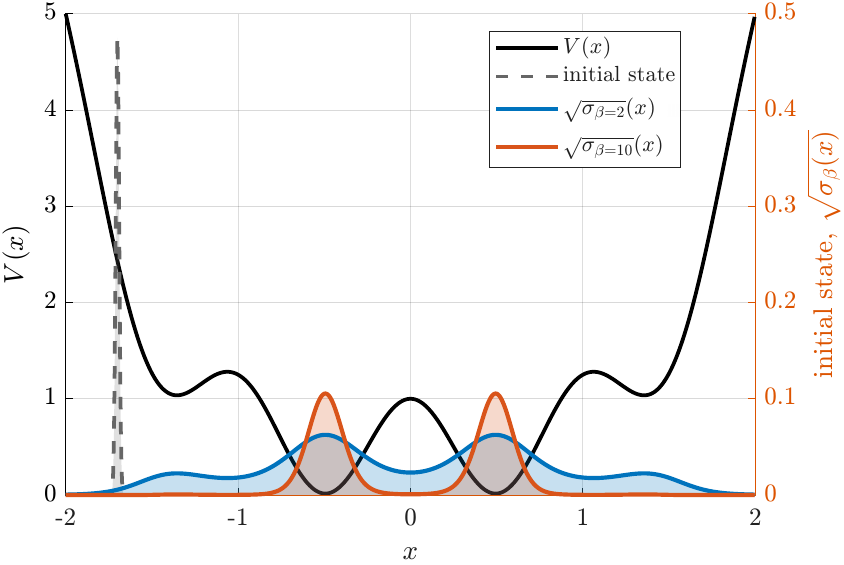}
        \caption{Initial state and encoded Gibbs states}
    \label{fig:warminit}
    \end{subfigure}
    \hfill
    \begin{subfigure}[t]{0.52\linewidth}
        \centering
        \includegraphics[width=\linewidth]{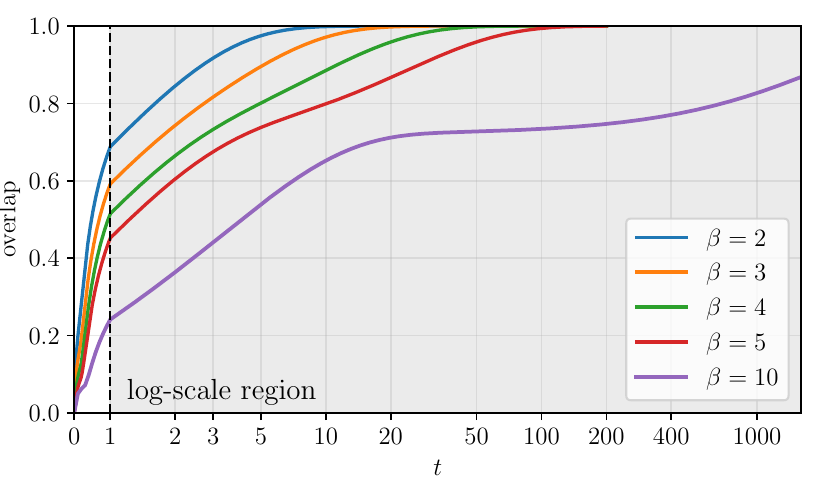}
        \caption{Overlap with respect to time}
    \label{fig:warmVsMix}
    \end{subfigure}
    \caption{Lindbladian-based warm-start generation. Left: Initial state for the Lindblad dynamics, and encoded Gibbs states at $\beta=2$ and $\beta=10$. Right: The overlap between the quantum state evolved under the Lindblad dynamics and the target Gibbs state, defined as in~\cref{eqn:mixed-state-overlap}. A short evolution time ($t \le 1$) is enough to prepare warm-start states with constant ($\ge 0.2$) overlap even at low temperatures.}
    \label{fig:combined-fig-lindblad-warm-start}
\end{figure}

\section*{Discussion}
In this work, we present a novel framework for accelerating probabilistic sampling using quantum computers. By associating the Gibbs sampling task with the preparation of the encoded Gibbs state, we leverage quantum singular value thresholding to accelerate classical sampling. Our method operates directly at the operator level, without requiring explicit time discretization of stochastic processes, thereby circumventing existing classical algorithms based on simulating discrete-time Markov chains. To the best of our knowledge, this provides the first provable quantum speedup for a broad class of sampling processes, directly in terms of the spectral gap of the infinitesimal generator, which is a type of speedup with no classical counterpart.

Our quantum algorithm relies on preparing the encoded Gibbs state $\ket{\sqrt{\sigma}}$, which corresponds to the stationary state of a Markov chain (i.e., a forward Kolmogorov equation). In the literature, the quantum state corresponding to the stationary distribution of a classical Markov chain is often referred to as a ``qsample.'' Generating a qsample is a long-standing challenge in quantum computing and is widely regarded as a significantly more powerful task than its classical counterpart~\cite{aharonov2003adiabatic,szegedy2004quantum,temme2025quantized}. In our algorithm, the state $\ket{\sqrt{\sigma}}$ is fully characterized as the ground state of the (generalized) Witten Laplacian $\mathcal{H}$.
In fact, given any target probability distribution with density $\rho$, as long as the quantum state $\ket{\sqrt{\rho}}$ can be characterized as an eigenstate of a linear operator, a similar singular value thresholding approach may apply to facilitate the sampling from $\rho$. For example, the ground state of the index-1 Witten Laplacian (the standard Witten Laplacian is also called the index-0 Witten Laplacian) encodes a distribution that locates the metastable configurations over a non-convex potential function~\cite{HelfferNier2004,lelievre2024using}. In such cases, our framework may be extended to enable sampling from non-Gibbs distributions, for which few efficient classical algorithms are known.

Our construction of block-encodings for the factor operators $L_j$ relies on the quantum implementation of the Fourier transform in $\RR^d$. Extending these techniques to more general manifolds or to non-differentiable settings remains an open question and an interesting avenue for future research.

In the limit of zero temperature (i.e., $\beta \to \infty$), non-logconcave sampling reduces to a non-convex optimization problem. Recent works~\cite{leng2023quantum,chen2025quantum,leng2025quantum} propose to leverage quantum dynamics (e.g., Hamiltonian or Lindbladian evolution) for solving non-convex optimization problems. These dynamics-based quantum optimization algorithms operate independently of classical processes and can exhibit super-polynomial speedups over all classical algorithms for some problem classes~\cite{leng2025sub}. This raises a natural question: can certain quantum dynamics be harnessed to achieve large (i.e., super-quadratic) acceleration for non-logconcave sampling? If such methods exist, they would greatly expand the design space of quantum algorithms with practical utility.

Quantum Gibbs sampling, which prepares a thermal state $\rho \propto e^{-\beta H}$ for a quantum Hamiltonian $H$, has seen rapid progress in algorithmic developments and analysis in recent years~\cite{RallWangWocjan2023,chen2023quantum,ding2025efficient,gilyen2024quantum,BardetCapelGaoEtAl2023,rouz2024,rouze2024optimal,kochanowski2024rapid,tong2024fast}. Since classical Gibbs sampling corresponds to the special case where $H$ is diagonal in the computational basis,  we hope that the interplay between classical and quantum perspectives may help uncover further quantum advantages in Gibbs sampling tasks.

\section*{Acknowledgments}
This work is partially supported by the Simons Quantum Postdoctoral Fellowship, DOE QSA grant \#FP00010905 (J.L.), and a Simons Investigator Award in Mathematics through Grant No. 825053 (J.L., L.L.). Support is also acknowledged from the U.S. Department of Energy, Office of Science, National Quantum Information Science Research Centers, Quantum Systems Accelerator (Z.D., Z.C.) and the U.S. Department of Energy, Office of Science, Accelerated Research in Quantum Computing Centers, Quantum Utility through Advanced Computational Quantum Algorithms, grant no. DE-SC0025572 (L.L.). We thank David Limmer, Jianfeng Lu, Lexing Ying and Ruizhe Zhang for helpful suggestions.

\section*{Author Contributions}
J.L., Z.D., and L.L. conceived the original study and carried out the theoretical analysis to support the study. J.L. and Z.C. carried out numerical calculations to support the study. All authors, J.L., Z.D., Z.C., and L.L. discussed the results of the manuscript and contributed to the writing of the manuscript.

\bibliographystyle{myhamsplain}
\bibliography{refs}

\newpage
\appendix
\begin{center}
    {\huge Appendices}
\end{center}
\tableofcontents
\setcounter{figure}{0}
\renewcommand{\thefigure}{S\arabic{figure}}

\section{Mathematical Preliminaries on Markov Processes}\label{append:math-prelim}

\paragraph{Notation.}
We consider two types of inner products. 
The first is the standard $L^2$ inner product $\langle f, g \rangle \coloneqq \int_{\mathbb{R}^d} f g \d x$, where we omit the argument $x$ for simplicity. The second is the $\sigma$-weighted inner product, defined as $\langle f, g \rangle_\sigma \coloneqq \int_{\mathbb{R}^d} \left(f g\right) \sigma\d x$. We also define $\norm{f}^2_{L^2}=\langle f, f \rangle$ and $\norm{f}^2_\sigma=\langle f, f \rangle_\sigma$.
For any operator $\mc{L}$ acting on a dense subset of $L^2(\RR^d)$, let $\mathcal{L}^\dagger$ be the adjoint with respect to the standard $L^2$ inner product, i.e., $\langle f, \mc{L}^{\dag}g \rangle=\langle \mc{L} f, g \rangle$ for all proper functions $f, g$.

\vspace{4mm}
In this work, for simplicity, we only consider absolutely continuous probability measures that admit a density, and \emph{ergodic} dynamics with a unique fixed point $\sigma$. Starting from certain initial distribution $\rho(0)$ from a properly chosen set $\mathcal{S}$, the \emph{mixing time} is defined as: 
\begin{equation}
    t^{[\cdot]}_{\rm mix}(\epsilon) = \inf_{t\ge 0} \sup_{\rho(0)\in \mc{S}}\left\{ [\cdot] \left( \rho(t), \sigma \right) \leq \epsilon \right\}.
\end{equation}
Here, $[\cdot]$ measures the discrepancy between the densities $\rho$ and $\sigma$. 
A standard choice is the total variation (TV) distance~\cite{Leo_1992,billing_1999}:
\begin{equation}\label{eqn:tv-distance-defn}
    \mathrm{TV}(\rho, \sigma) = \frac{1}{2}\int_{\RR^d} |\rho(x) - \sigma(x)|~\d x.
\end{equation}
Besides the TV distance, $[\cdot]$ can be taken to be the $\chi^2$-divergence (see below), Wasserstein-$2$ ($W_2$) distance~\cite{vaserstein1969markov,Kantorovich_1960}, among others~\cite{Morimoto_1963,renyi1961measures,Ali_1966,Csiszar_1975}.

Let $\mathcal{L}$ be the infinitesimal generator of a continuous-time process. We say the process satisfies the detailed balance condition~\cite{vanKampen1992stochastic,ohagan2004kendall} if $\mathcal{L}^\dagger$ is self-adjoint with respect to the $\sigma$-weighted inner product, i.e.,
\begin{equation}\label{eqn:detailed_balance}
\langle f, \mathcal{L}^\dagger g \rangle_\sigma = \langle \mathcal{L}^\dagger f, g \rangle_\sigma.
\end{equation} 
For instance, for the Fokker--Planck equation \cref{eqn:FKPK}, we have $\mc{L}^{\dag}=-\nabla V\cdot \nabla+ \beta^{-1} \Delta$. Direct calculation shows that $\mc{L}^{\dag}$ satisfies the detailed balance condition.

When the detailed balance condition is satisfied, all the eigenvalues of $\mathcal{L}^\dagger$ are real and furthermore non-positive. Since $\mc{L}^{\dag}(1)=0$ always holds, the process is ergodic if and only if $\ker(\mathcal{L}^\dagger)$ is one-dimensional~\cite{villani2009hypocoercivity}. For an ergodic dynamics, all other eigenvectors of $\mc{L}^{\dag}$ are orthogonal to $1$ with respect to the $\sigma$-weighted inner product. Thus the \emph{spectral gap} of $\mc{L}^{\dag}$ has the following variational characterization 
\begin{equation}\label{eq:abstract-L-gap}
\mathrm{Gap}\left(\mathcal{L}^\dagger\right):=\inf_{f\notin \ker(\mathcal{L}^\dagger)}\frac{\langle f, -\mathcal{L}^\dagger f \rangle_\sigma}{\|f-\int f \sigma \d x\|^2_{\sigma}}.
\end{equation} 

Define the variance of a function $f$ as $\mathrm{Var}_\sigma(f):=\|f-\int f \sigma \d x\|^2_{\sigma}$. Let $f(x,t)=\rho(x,t)/\sigma(x)$. Then 
the discrepancy between $\rho(t)$ and $\sigma$ can be measured by the variance $\mathrm{Var}_\sigma(f)=\|f-1 \|^2_{\sigma}=\chi^2(\rho(t),\sigma)$, where  $\chi^2(\rho,\sigma)=\|\rho/\sigma-1\|^2_{\sigma}$ is called the $\chi^2$-divergence and we have used $\int f \sigma \d x=1$. Notice that 
\begin{equation}\label{eqn:var_reduction}
    \partial_t \mathrm{Var}_\sigma(f)=2\langle f-1, \partial_t f\rangle_\sigma=2\langle f-1, \mathcal{L}(f\sigma)\rangle=-2\langle- \mc{L}^\dagger(f), f\rangle_\sigma\le -2 \mathrm{Gap}\left(\mathcal{L}^\dagger\right)  \mathrm{Var}_\sigma(f).
\end{equation}
If the spectral gap is positive, then $\rho(t)$ converges exponentially to $\sigma$ in $\chi^2$-divergence:
\begin{equation}\label{eqn:chi_square_convergence}
\chi^2(\rho(t),\sigma)\leq \exp\left(-2  \mathrm{Gap}\left(\mathcal{L}^\dagger\right)  t\right)\chi^2(\rho(0),\sigma).
\end{equation}
The last inequality in \cref{eqn:var_reduction} is also called the \emph{Poincar\'e inequality} with a \emph{Poincar\'e constant} $C_{\rm PI}=1/\mathrm{Gap}\left(\mathcal{L}^\dagger\right)$~\cite{villani2009optimal,ledoux2001concentration,bakry2014analysis}. 

The Poincar\'e inequality immediately leads to the exponential convergence of $\rho(t)$ to the stationary measure $\sigma$ in $\chi^2$-divergence~\cite[Theorem 1.2.21]{chewi2023log}:
\begin{equation}
    \chi^2(\rho(t),\sigma) \le e^{-\frac{2}{C_{\rm PI}}t}\chi^2(\rho(0),\sigma).
\end{equation}

Therefore, a large spectral gap (i.e., a small Poincar\'e constant) together with a mild initial $\chi^2$-divergence $\chi^2(\rho(0),\sigma)$ implies a fast convergence (i.e., mixing of the process). On the other hand, a small spectral gap is a strong indicator that the mixing process can be slow.

\section{Related works}\label{sec:Re}
In this section, we briefly review previous works on continuous sampling in both quantum and classical literature. For simplicity, we focus on a target distribution $\sigma \propto e^{-V}$, i.e., we set $\beta=1$. A comparison of different results is drawn in~\cref{table:comparison}.

Solving continuous sampling problems on a quantum computer is a relatively recent development. The current state-of-the-art approaches are based on quantization of classical sampling algorithms using quantum walks \cite{childs2022quantum,ozgul2024stochastic,ozgul2025quantumspeedupsmarkovchain}. Ref. \cite{childs2022quantum} proposes several quantum algorithms based on classical sampling algorithms such as ULA and MALA, and achieves quadratic speedup for logconcave distributions. Specifically, assuming the potential $V$ is $\gamma$-strongly convex and a warm start initial state, the complexity of the quantum MALA algorithm is $\mathcal{O}(d^{1/4}/\gamma^{1/2})$~\cite{childs2022quantum}. The quantum algorithm developed in Ref.~\cite{ozgul2024stochastic} (also termed quantum MALA)  is applicable to a broader class of non-logconcave distributions. Given a non-logconcave distribution with density $\sigma$, the Cheeger constant~\cite{cheeger1969lower,levin2017markov} is defined as
\begin{equation}\label{eqn:cheeger}
C_{\rm CG}=\inf_{A\in\mathbb{R}^d}\frac{\liminf_{h\rightarrow0^+}\frac{1}{h}\int_{A_h\setminus A}\sigma\mathrm{d} x}{\min\left\{\int_{A}\sigma\mathrm{d} x,\int_{A^c}\sigma\mathrm{d}x\right\}}\,,
\end{equation}
where $A_h=\left\{x \colon \exists y\in A,\ \|x-y\|\leq h\right\}$. 
Starting from a warm start initial state, the complexity is $\mathcal{O}(d^{1/2}C_{\rm CG})$. While the Cheeger constant can be lower bounded by the Poincar\'e constant following Cheeger's inequality $C_{\mathrm{PI}} \leq 4C^2_{\mathrm{CG}}$~\cite[Theorem 13.10]{levin2017markov}, a Busher type inequality states that when $V(x)$ is smooth and $\nabla V$ is $L$-Lipschitz, $C_{\rm PI}$ also produces an upper bound for $C_{\rm CG}$ as $\min\left\{C_{\mathrm{CG}}/(6\sqrt{L}),\,C_{\mathrm{CG}}^2/36\right\}\leq C_{\mathrm{PI}} \leq 4C_{\mathrm{CG}}^2$ (\cite[Theorem 5.2]{ledoux2004spectral}). In particular, for a class of highly nonconvex potentials (such as double well potentials), both $C_{\mathrm{CG}}$ and $C_{\mathrm{PI}}$ can be large and  $C_{\mathrm{PI}}=\wt{\Theta}(C_{\mathrm{CG}})$. In this case, our complexity achieves a quadratic improvement in terms of the Poincar\'e constant compared to the results presented in~\cite{ozgul2024stochastic}. Ref.~\cite{ozgul2025quantumspeedupsmarkovchain} mainly focus on improving the low-accuracy quantum sampler and achieves the quantum speed up when applied to optimization.    

In Refs.~\cite{childs2022quantum,ozgul2024stochastic}, the warm-started initial state can be obtained using annealing techniques~\cite[Section B.3]{childs2022quantum},~\cite[Section 3]{ozgul2024stochastic}. Specifically, when the potential $V$ is $\gamma$-strongly convex, the complexity of Quantum Annealed MALA is $\mathcal{O}(d/\gamma^{1/2})$~\cite[Theorem C.7]{childs2022quantum}. When the target distribution satisfies the Log-Sobolev inequality with constant $C_{\rm LSI}$~\footnote{We note that the Log-Sobolev is stronger than the Poincar{\'e} inequality; see the remarks under \cref{table:comparison}. For logconcave distributions, these inequalities are nearly equivalent~\cite{milman2009role}.}, the complexity becomes $\mathcal{O}(d C_{\rm LSI} C_{\rm CG})$~\cite[Theorem 6]{ozgul2024stochastic}. The annealing techniques can also be used in our method to achieve similar complexity when the target distribution is logconcave or satisfies the Log-Sobolev inequality.

For comparison, continuous Gibbs samplers on classical computers can be broadly classified into two categories: low-accuracy samplers, whose complexity scales inversely polynomially with the precision $\epsilon$, and high-accuracy samplers, whose complexity scales as $\mathrm{polylog}(1/\epsilon)$. Low-accuracy samplers typically arise from discretizations of stochastic processes whose stationary distributions converge to the target distribution $\sigma$, such as overdamped or underdamped Langevin dynamics~\cite{JKO98,Wibisono2018,MA_2021}. In our work, we focus specifically on high-accuracy samplers; therefore, we omit a detailed review of complexity results for low-accuracy samplers. Interested readers are referred to~\cite{Shen_2019,Cao_2021,ChewiErdogduLiEtAl2024,chewi2023log} for a detailed review.
High-accuracy samplers, such as MALA and MHMC, often incorporate a Metropolis–Hastings correction step~\cite{Bes_95,RR98}, which ensures the algorithm is unbiased. This correction step enables the use of longer time steps in the algorithm, which reduces the complexity of the algorithm with respect to the precision parameter~\cite{Dwivedi_2018,Chen_2020,Chewi2020OptimalDD,Lee_2020,Wu_2022,chen2023,Chewi_2024}. Most of the analysis of these high-accuracy samplers requires a warm start assumption and/or strongly log-concavity. Assuming warm start assumption and $V$ is $\gamma$-strongly convex, MALA achieves $\epsilon$ accuracy in TV, $\sqrt{\mathrm{KL}}$, $\sqrt{\chi^2}$, $\sqrt{\gamma}W_2$ distance with complexity scaling as $\widetilde{\Or}\left(d^{1/2}/\gamma\right)$, which is proved to be optimal for MALA~\cite{Wu_2022}. Relaxing the $\gamma$-strongly convex condition to isoperimetric bounds in the case of MALA become more complicated, and the question of deriving analogous results for the Poincar\'e constant remains open, to the best of our knowledge. For general isoperimetric bounds, instead of considering MALA,~\cite[Chapter 8.6]{chewi2023log} prove that using the proximal sampler~\cite{Titsias_2018,Lee_2021,Chen_2022}, it is possible to achieve $\epsilon$-precision in various distance with complexity scaling as $\widetilde{\Or}\left(d^{1/2}C_{\rm ISB}\right)$, where $C_{\rm ISB}$ is the Log-Sobolev or Poincar\'e constant. \cite[Lemma 6.5]{osti_10276674} analyzes the complexity of projected MALA, where sampling is constrained to a bounded domain with radius independent of $d$. 

In terms of different assumptions, we note that if a function $V$ is $\gamma$-strongly convex, then $\exp(-V)$ has $1/\gamma$-Log-Sobolev constant. In addition, Log-Sobolev inequality and Cheeger's inequality are stronger than Poincar{\'e} inequality. Specifically, if a distribution satisfies Log-Sobolev inequality with parameter $C_{\rm LSI}$ then this distribution also satisfies Poincar\'e inequality with the same constant $C_{\rm LSI}$. Furthermore, if a distribution satisfies Cheeger's inequality with parameter $C_{\rm CG}$ then this distribution also satisfies Poincar\'e inequality with the same constant $4C_{\rm CG}^2$~\cite{cheeger1969lower}. For the classical samplers, we note that by combining a low-accuracy sampler, based on underdamped Langevin Monte Carlo and a proximal sampler with MALA, and assuming access to a warm stationary point of $V$, the MALA algorithm in~\cite[Theorem 5.4]{Chewi_2024} achieves the same scaling with the $C_{\rm PI}$-Poincar{\'e} constant, without requiring a warm start for initialization. Related developments for a particular class of non-convex potentials called Gaussian mixture models can be found in~\cite{PB_2020,DongTong2022}.

\section{Ground State preparation with $\mathcal{H} = \sum_j L^\dagger_j L_j$}\label{append:meta-algorithm-qsvt}

In this section, we introduce a quantum algorithm that prepares the ground state of a quantum Hamiltonian of the form $\mathcal{H} = \sum_j L^\dagger_j L_j$. In our algorithm, this operator $\mathcal{H}$ represents the (spatially discretized) generalized Witten Laplacian.

\subsection{Problem formulation}
Suppose that we have an $N$-by-$N$ operator of the form
\begin{align}\label{eqn:A-def}
    \mathcal{H} = \sum^p_{j=1} L^\dagger_j L_j,
\end{align}
where each $L_j\colon \CC^N \to \CC^N$ is a complex-valued matrix. Since $\mathcal{H}$ is Hermitian and non-negative, by the spectral theorem, all eigenvalues of $\mathcal{H}$ are non-negative real numbers.
We denote the eigenvalues of $\mathcal{H}$ as $0 \le \lambda_1 < \lambda_2 \le \dots \le \lambda_N$.
Moreover, we assume that there is a positive $D > 0$ such that
\begin{align}\label{eqn:H-abstract-spectra}
    0 \le \lambda_1 < \frac{D}{16},\quad \lambda_k \ge \frac{9D}{16} \quad \forall k=1,2,\dots
\end{align}
In other words, the spectral gap (i.e., the difference between the first two eigenvalues) of $\mathcal{H}$ is of the order $\Theta(D)$.
The eigenvector of $\mathcal{H}$ associated with the smallest eigenvalue $\lambda_1$ is referred to as the \textit{ground state} of $\mathcal{H}$, denoted by $\ket{g}$. We can write $\mathcal{H} = \mathbb{L}^\dagger\mathbb{L}$, where the block matrix
\begin{equation}\label{eqn:L-abstract}
    \mathbb{L} \coloneqq [L^\top_1, L^\top_2,\dots, L^\top_d]^\top.
\end{equation}

\begin{lem}\label{lem:eigen-singular}
    Suppose the matrix $\mathbb{L}$ has singular values $\{\sigma_k\}^N_{k=1}$, arranged in an ascending order.
    Then, we have
    \begin{align}\label{eqn:singular-value-bound}
        0 \le \sigma_1 \le \frac{\sqrt{D}}{4},\quad \sigma_k \ge \frac{3\sqrt{D}}{4}\quad \forall k = 2,\dots,N.
    \end{align}
    Moreover, the right singular vector of $\mathbb{L}$ associated with the  $\sigma_1$ is $\ket{g}$.
\end{lem}
\begin{proof}
    Suppose that the matrix $\mathbb{L}$ has a singular value decomposition (SVD) as $\mathbb{L} = U \Sigma V^\dagger$, then we have $\mathcal{H} = V \Sigma^2 V^\dagger$. The ground state of $\mathcal{H}$ corresponds to the first column in $V$, which is the singular vector of $\mathbb{L}$ associated with the smallest singular value.
    Also, there is an one-to-one correspondence between the singular values of $\mathbb{L}$ and the eigenvalues of $\mathcal{H}$: $\sigma_k = \sqrt{\lambda_k}$. Therefore,~\cref{eqn:singular-value-bound} is a direct consequence of~\cref{eqn:H-abstract-spectra}.
\end{proof}

In quantum numerical linear algebra, the block-encoded matrix is a standard input model that enables several powerful quantum algorithms, including Quantum Singular Value Transformation (QSVT). We now define the block-encoding of a rectangular matrix $A \in \CC^{2^n\times 2^p}$. We assume $n\ge p$ without loss of generality.

\begin{defn}[Block-encoding of a rectangular matrix]\label{defn:be-general}
    Given a matrix $A\in \CC^{2^n\times 2^p}$ with $n\ge p$, if we can find $\alpha, \epsilon > 0$, and a unitary matrix $U_A\in \CC^{2^{n+m}\times 2^{n+m}}$ such that
    \begin{equation}
        \|A - \alpha \left(\bra{0^m}\otimes I_{2^n}\right) U_A \left(\ket{0^m}\otimes I_{2^p}\right)\| \le \epsilon,
    \end{equation}
    then $U_A$ is called an $(\alpha, m, \epsilon)$-block-encoding of $A$. 
\end{defn}

Intuitively, the unitary operator $U_A$ encodes the matrix $A/\alpha$ in its upper left corner, up to an additive error $\epsilon$. 
The parameter $\alpha$ is referred to as the \textit{normalization factor}, which ensures the block-encoded matrix $A/\alpha$ has an operator norm no greater than $1$.
When the precision of the block encoding can be easily controlled, for simplicity we may set $\epsilon=0$. In this case, we have $A = \alpha \left(\bra{0^m}\otimes I_{2^n}\right) U_A \left(\ket{0^m}\otimes I_{2^p}\right)$, and $U_A$ is called an $(\alpha, m)$-block-encoding of $A$. We will extensively use the block encoding of $\mathbb{L}$.

The goal is to prepare a ground state of $\mathcal{H}$ using a quantum computer. Thanks to the factorization $\mathcal{H} = \mathbb{L}^\dagger\mathbb{L}$, the problem is equivalent to preparing the right singular vector of $\mathbb{L}$ associated with the smallest singular value $\sigma_1$. This can be achieved by applying a singular value thresholding (SVT) algorithm to an initial state $\ket{\phi}$ to filter out contributions in $\ket{\phi}$ corresponding to higher singular values of $\mathbb{L}$. When the initial state has an $\Omega(1)$ overlap with the target state $\ket{g}$, we will end up with the desired ground state with a constant success probability. We present the complexity analysis of the quantum algorithm in Appendix~\ref{append:proof-main-qsvt}. The singular value thresholding algorithm is implemented by QSVT, which is discussed in the next subsection.

\subsection{Singular value thresholding via QSVT}

In this section, we provide a brief introduction to the Quantum Singular Value Transformation (QSVT) algorithm\cite{gilyen2019quantum} for completeness.

For a matrix $A \in \CC^{2^n\times 2^p} (n\ge p)$, we consider its singular value decomposition:
\begin{equation}
    A = W \Sigma V^\dagger.
\end{equation}
The columns of $W, V$ are called the left and right singular vectors of $A$, respectively. $\Sigma$ is a $2^n\times 2^p$ matrix with the main diagonal given by the singular values $\{\sigma_1, \cdots, \sigma_{2^p}\}$. When $A$ is given as a block encoding, we must have $\|A\| \le 1$ and the singular values of $A$ are in the interval $[0,1]$.

Let $f\colon \RR \to \CC$ be a scaler function such that $f(\sigma_j)$ is well-defined for all singular values $\sigma_j$, we can define a \textit{right} generalized matrix function:
\begin{equation}
    f^{\triangleright}(A) = V f(\Sigma) V^\dagger,
\end{equation}
where $f(\Sigma) = \diag(f(\sigma_1),\dots, f(\sigma_N))$. Similarly, we can define a \textit{left} matrix function and a \textit{balanced} matrix function induced by $f$. They will not appear in this work. 

When the function $f$ is specified as an even polynomial of degree $d$, the matrix function $ f^{\triangleright}(A)$ can be implemented on a quantum computer via QSVT. 
For a square matrix $A$, QSVT can be implemented following \cite[Corollary 11]{gilyen2019quantum}. This result can be generalized to rectangular matrices (see e.g.,~\cite{tang2024cs}). The following theorem is adapted from~\cite[Theorem 2.3]{tang2024cs}.

\begin{thm}[QSVT with even polynomials]\label{thm:qsvt}
    Let $A \in \CC^{2^n\times 2^p}$ be encoded by its $(\alpha,m)$-block-encoding $U_A$. For an even polynomial $F(x) \in \RR[x]$ with degree $d$ and $|F(x)|\le 1$ for any $x \in [-1,1]$, we can implement a $(1,m+1)$-block-encoding of $F^{\triangleright}(A/\alpha)$ using $U_A$, $U^\dagger_A$, $m$-qubit controlled NOT, and single-qubit rotation gates for $\mathcal{O}(d)$ times. 
\end{thm}

To implement the singular value thresholding algorithm for $A$, we want to filter out all singular values that are greater than or equal to $\sigma_2$. Suppose that $0 \le \sigma_1 \le s_1 < s_2 \le \sigma_2$, and we consider the following rectangular filter function
\begin{align}\label{eqn:rectangular-function}
    f(x) = \begin{cases}
        1, & x \in [-s_1,s_1],\\
        0, & x \in [-1, -s_2]\cup [s_2, 1].
    \end{cases}
\end{align}
Rectangle functions are commonly used in QSVT, and they can be efficiently approximated by even polynomials with an additive error.

\begin{lem}[{\cite[Corollary 16]{gilyen2019quantum}}]\label{lem:rectangle-polynomial}
    Let $\delta, \epsilon \in (0, 1/2)$ and $t \in [-1,1]$. There exist an even polynomial $P(x) \in \RR[x]$ of degree $\mathcal{O}(\delta^{-1}\log(\epsilon^{-1}))$, such that $|P(x)| \le 1$ for all $x \in [-1,1]$ and 
    \begin{align*}
        \begin{cases}
            P(x) \in [0, \epsilon] & x \in [-1, -t-\delta] \cup [t + \delta, 1],\\
            P(x) \in [1 - \epsilon, 1] & x \in [-t + \delta, t - \delta].
        \end{cases}
    \end{align*}
\end{lem}
This approximation is asymptotically optimal in the parameters $\delta$ and $\epsilon$~\cite{eremenko2006uniform}. In practice, the approximate polynomial can be explicitly constructed via convex optimization based methods (see~\cite[Section IV]{DongLinTong2022}). 
Based on Lemma~\ref{lem:rectangle-polynomial}, we can construct an even polynomial that approximates the filter function $f$, which leads to an efficient quantum implementation of the singular value thresholding algorithm~\cite[Theorem 19]{gilyen2019quantum}. In what follows, we provide a generalized version that applies to rectangular matrices, together with a proof illustrating the algorithmic procedure.

\begin{prop}[Singular value thresholding]\label{prop:sv-filter}
    Let $A \in \CC^{2^n\times 2^p}$ be encoded by its $(\alpha,m)$-block-encoding $U_A$. Let $\sigma_1, \sigma_2$ be the first two singular values of $A$ and $0 \le \sigma_1 \le s_1 < s_2 \le \sigma_2$. We denote $s = s_2 - s_1$.  
    Let $f(x)$ be the rectangular filter function given in~\cref{eqn:rectangular-function}.
    We can implement a $(1,m+1,\epsilon)$-block-encoding of the matrix function $f^{\triangleright}(A)$ using $U_A$, $U^\dagger_A$, $m$-qubit controlled NOT, and single-qubit rotation gates for $\mathcal{O}\left(\alpha s^{-1}\log(\epsilon^{-1})\right)$ times.
\end{prop}
\begin{proof}
    Note that $f^{\triangleright}(A)$ is equivalent to $f^{\triangleright}_{1/\alpha}(A/\alpha)$, where $f_{1/\alpha}(x)$ is a filter function given in~\cref{eqn:rectangular-function} with $s_1$ and $s_2$ scaled by a factor of $1/\alpha$.
    By choosing $t = (s_1+s_2)/(2\alpha)$ and $\delta = (s_2-s_1)/(2\alpha)$ in Lemma~\ref{lem:rectangle-polynomial}, we can approximate the filter function $f_{1/\alpha}(x)$ up to an additive error $\epsilon$ with an even polynomial $P(x)$ of degree $d = \mathcal{O}(\alpha s^{-1}\log(\epsilon^{-1}))$.
    Then, by Theorem~\ref{thm:qsvt}, we can implement a $(1,m+1)$-block-encoding of $P^{\triangleright}(A/\alpha)$ using $U_A$, $U^\dagger_A$, $m$-qubit controlled NOT, and single-qubit rotation gates for $\mathcal{O}(\alpha s^{-1}\log(\epsilon^{-1}))$ times. This yields is a $(1,m+1,\epsilon)$-block-encoding of $f^{\triangleright}(A)$ since $|P(x) - f_{1/\alpha}(x)| \le \epsilon$ for all $x \in [-1,1]$.
\end{proof}

\subsection{Singular value thresholding for ground state preparation}\label{append:proof-main-qsvt}
\begin{thm}\label{thm:main-qsvt}
    Consider $\mathcal{H} = \mathbb{L}^\dagger \mathbb{L}$ satisfying ~\cref{eqn:H-abstract-spectra} with $\mathbb{L}$ given by~\cref{eqn:L-abstract}.
    Suppose that we have access to \textit{(i)} a $(\alpha, m)$-block-encoding of $\mathbb{L}$ (denoted by $\UU$), and \textit{(ii)} a warm start quantum state $\ket{\phi}$ with constant overlap with the ground state $\ket{g}$ of $\mathcal{H}$, i.e., $\abs{\braket{\phi}{g}} = \Omega (1)$.
    Then for any $\varepsilon > 0$, we can prepare a quantum state $\ket{\tilde{g}}$ such that $\|\tilde{g} - g\| \le \varepsilon$ using $\mathcal{O}(\alpha D^{-1/2}\log(\varepsilon^{-1}))$ queries to $\UU$ and $\Or(1)$ copies of the state $\ket{\phi}$.
\end{thm}

\begin{proof}
    By Lemma~\ref{lem:eigen-singular}, the ground state of $\mathcal{H}$ is the same as the right singular vector of $\mathbb{L}$ associated with the smallest singular value $\sigma_1\le \sqrt{D}/4$.
    Therefore, we can implement a singular value thresholding by filtering out all (right) singular vectors associated with singular values greater than or equal to $3\sqrt{D}/4$. 
    Let $f(x)$ be the rectangular filter function defined in~\cref{eqn:rectangular-function} with $s_1 = \sqrt{D}/4$ and $s_2 = 3\sqrt{D}/4$. We can apply Proposition~\ref{prop:sv-filter} to construct a $(1, m+1, \epsilon)$-block-encoding of $f^{\triangleright}_s(\mathbb{L})$ with $\mathcal{O}(\alpha D^{-1/2}\log(\epsilon^{-1}))$ queries to $\UU$. In other words, this quantum circuit block-encodes an approximate projector $\tilde{\Pi}$ such that $\|\tilde{\Pi} - \ketbra{g}{g}\| \le \epsilon$.
    By applying this quantum circuit to an initial state $\ket{0^{m+1}}\otimes\ket{\phi}$, we will obtain a quantum state
    \begin{equation}
       \ket{0^{m+1}}\ket{\tilde{g}} + \ket{\perp},
    \end{equation}
    where $\tilde{g} = \tilde{\Pi}\ket{\phi}$. It is deduced from $\abs{\braket{g}{\phi}} = \Omega(1)$ that $\|\tilde{g} - g\| \le \epsilon$ and $\|\tilde{g}\| = \Omega(1)$.
    Therefore, by post-selecting the quantum state flagged by $0^{m+1}$ (with success probability equals to $\|\tilde{g}\|^2 = \Omega(1)$), we can obtain an $\epsilon$-approximate ground state $\ket{\tilde{g}}$.
\end{proof}

\section{Operator-Level Acceleration of Langevin dynamics: Details}\label{append:langevin-complexity-analysis}

The Witten Laplacian of Langevin dynamics takes the form:
\begin{equation}\label{eqn:spatial_discretize_Lj}
    \mathcal{H} = \sum^d_{j=1}L^\dagger_j L_j,\quad L_j \coloneqq -i\frac{1}{\sqrt{\beta}} \partial_{x_j} - i\frac{\sqrt{\beta}}{2}\partial_{x_j}V \quad \forall j \in [d]=\{1,2,\dots,d\}.
\end{equation}
In this section, we use the expression $\mathcal{H} = \mathbb{L}^\dagger\mathbb{L}$ with $\mathbb{L} \coloneqq [L^\top_1, L^\top_2,\dots, L^\top_d]^\top$.

\subsection{Spatial discretization}\label{append:spatial_discretize_Lj}

For $j = 1,\dots, d$, the operators $L_j$ are unbounded operators defined in the real space $\RR^d$. To compute these operators on a quantum computer, we need to perform spatial discretization and map them to finite-dimensional matrices. 

Suppose that the Gibbs measure $\sigma$ has a negligible probability mass outside a box $\Omega = [-a, a]^d \subset \RR^d$. For the potential $V$ with sufficient growth rate, we may set $a = \mathcal{O}(\log(d/\epsilon))$.\footnote{For example, if $V \ge \gamma \|x\|^2$ for some $\gamma > 0$, then we have $\mathbb{P}_{\sigma}[x \notin \Omega] \le \frac{1}{Z}\int_{\RR^d \backslash \Omega }e^{-\gamma\|x\|^2}\d x = \mathcal{O}(dC^{-a})$, where $C > 1$ independent of $d$. Therefore, to truncate the Gibbs measure up to an error $\epsilon$, it is sufficient to choose $a = \mathcal{O}(\log(d/\epsilon))$.}
Since the Gibbs measure is effectively zero on the boundary of $\Omega$, we assume the numerical domain $\Omega$ is induced with a periodic boundary. 
For simplicity, we assume the problem is further rescaled and translated to the unit box $\Omega = [0,1]^d$, which will only incur an $\mathcal{O}(a) = \mathcal{O}(\log(d/\epsilon))$ multiplicative overhead in the normalization factor of the block-encoding of $\mathbb{L}$.
Let $\mathcal{M} = \{\bm{x} = (x_1,\dots, x_d)\colon x_j \in \{0, h, \dots, 1-h\}\}$ be a regular mesh in $\Omega$, where $h = 1/N$ and $N$ is the number of grid points used for each dimension.
In what follows, we will refer to $a$ as the \textit{truncation length} and $N$ the \textit{discretization number}.

\paragraph{Discretization of $L_j$.}
Each component operator $L_j$ can be represented as a \textit{pseudo-differential} operator  
\begin{align}
    L_j(u)(\bm{x}) = \int_{\RR^d} e^{2\pi i \bm{x}\cdot \bm{\xi}}a_j(\bm{x},\bm{\xi}) \hat{u}(\bm{\xi}) d \bm{\xi},\quad a_j(\bm{x},\bm{\xi}) = \frac{2\pi}{\sqrt{\beta}}\xi_j - i\frac{\sqrt{\beta}}{2} \partial_{x_j}V(\bm{x}) \quad \forall j \in [d],
\end{align}
where $a_j(\bm{x},\bm{\xi})$ is called the \textit{symbol} of $L_j$ and $\hat{u}(\bm{\xi})$ is the Fourier transform of $u(\bm{x})$. 
For every $j\in [d]$, the symbol $a_j(\bm{x},\bm{\xi})$ naturally splits into two components, where $2\pi\xi_j$ corresponds to the differential operator $-i\partial_{x_j}$, and $-i\beta \partial_{x_j}V/2$ is a multiplicative operator in the computational basis.

By truncating the Euclidean space $\RR^d$ to the numerical domain $\Omega=[0,1]^d$ and imposing a periodic boundary condition, we can represent the truncated operator as:
\begin{align}\label{eqn:pseudp-spec-Lj}
    \hat{L}_{j}(u)(\xx) = \mathcal{F}^{-1}\left(\frac{2\pi}{\sqrt{\beta}}\xi_j\mathcal{F}(u)(\bm{\xi})\right)(\xx) + \partial_{x_j}V(\xx) u(\xx), \quad \xx \in \Omega = [0,1]^d.   
\end{align}
Here, $\mathcal{F}$ and $\mathcal{F}^{-1}$ represents the Fourier transform and inverse Fourier transform in $\Omega$, respectively:
\begin{equation*}
    \mathcal{F}(u)(\bm{\xi}) = \int_\Omega u(\xx)e^{-2\pi i \bm{\xi}\cdot \xx} \d \xx, \quad \xi = (\xi_1,\dots, \xi_d) \in \ZZ^d 
\end{equation*}
\begin{equation*}
    \mathcal{F}^{-1}(w)(\xx) = \sum_{\bm{\xi}\in \ZZ^d} w(\bm{\xi})e^{2\pi i \bm{\xi}\cdot \xx}, \quad \xx \in [0,1]^d.
\end{equation*}
The Fourier transformation can be approximated by a quadrature over the uniform mesh $\mathcal{M}$, and then carried out via the Discrete Fourier Transform (DFT). We denote $\mathcal{K}$ as an indexing set: $\mathcal{K} = \{\bm{k} = (k_1,\dots, k_d)\colon 0 \le k_j \le N-1$, $j \in [d]\}$, and the Fourier transform can be approximately computed through the $d$-dimensional DFT:
\begin{equation}\label{eqn:d-dim-dft}
    \mathcal{F}(u)(\bm{\xi}) \approx \frac{1}{N^{d/2}}\sum_{\bm{k}\in \mathcal{K}} u(\xx_{\bm{k}}) e^{-2\pi i \xx_{\bm k}\cdot \bm{\xi}}.
\end{equation}
In the 1-dimensional case, DFT is equivalent to the Quantum Fourier Transform (QFT), given by the following unitary operator:
\begin{equation}
    U_{\rm FT}\ket{j} = \frac{1}{\sqrt{N}}\sum^{N-1}_{k=0} e^{2\pi i kj/N}\ket{k},
\end{equation}
which can be implemented on a quantum computer using $\mathcal{O}(\log^2(N))$ elementary gates and no ancilla qubits~\cite{coppersmith2002approximate}. In general, a $d$-dimensional DFT (denoted by $\mathrm{DFT}_{N,d}$) can be implemented by concatenating $d$ 1-dimensional DFT/QFT:
\begin{equation}
    \mathrm{DFT}_{N,d} = \underbrace{\mathrm{DFT}_{N,1} \otimes \dots \otimes \mathrm{DFT}_{N,1}}_{d~\text{copies}} = U_{\rm FT}^{\otimes d},
\end{equation}
which can be implemented on a quantum computer using $\mathcal{O}(d\log^2(N))$ elementary gates.

Based on the pseudo-differential operator representation~\cref{eqn:pseudp-spec-Lj}, we can write down the spatial discretization of the operator $L_j$ as an $N^d$-dimensional matrix $\tildeL_j$:
\begin{equation}\label{eqn:dft-discretized-Lj}
    \tildeL_j = \frac{2\pi}{\sqrt{\beta}}\mathrm{DFT}^{\dagger}_{N,d} \left(\Phi^j_{{\bm \xi}}\right) \mathrm{DFT}_{N,d} - i\frac{\sqrt{\beta}}{2} \Phi^j_{\xx},
\end{equation}
where $\Phi^j_{\bm \xi}$ and $\Phi^j_{\xx}$ are diagonal matrices defined as follows:
\begin{equation*}
    \Phi^j_{\bm \xi}\ket{\xi_1,\dots,\xi_d} = \xi_j \ket{\xi_1,\dots,\xi_d},\quad \Phi^j_{\xx}\ket{x_1,\dots, x_d} = \partial_{x_j}V(\xx)\ket{x_1,\dots, x_d}.
\end{equation*}
They represent the multiplicative operators corresponding to $\xi_j$ and $\partial_{x_j}V$ in~\cref{eqn:pseudp-spec-Lj}, respectively. Note that the frequency number must be in the range $-N/2 \le \xi_j \le N/2-1$ for all $j \in [d]$. As a result, the discretized $\mathbb{L}$ operator is as follows:
\begin{equation}
    \widetilde{\mathbb{L}} = [\tildeL^\top_1,\dots, \tildeL^\top_d]^\top.
\end{equation}

\paragraph{Discretization of the encoded Gibbs state.}
Next, we discuss the kernel of the discretized operator $\widetilde{\mathbb{L}}$. As expected, it should correspond to a discretized encoded Gibbs state. Here, we adopt a standard trigonometric interpolation to characterize the discretization error. 
For $k = 0,\dots, N-1$, recall that the quadrature point $x_k = kh = k/N$, and we define the following function:
\begin{equation}
    \psi_k(x) = \frac{1}{\sqrt{N}}\sum^{N/2-1}_{j=-N/2} e^{2\pi ij(x-x_k)}, \quad x \in [0,1],
\end{equation}
which is a trigonometric polynomial satisfying $\psi_k(x_l) = \sqrt{N}\delta_{k,l}$ for any $k,l = 0,\dots,N-1$. Moreover, it is readily verified that 
\begin{equation}
    \langle \psi_k , \psi_l \rangle_{L^2} = \delta_{k,l} \quad \forall k,l = 0,\dots,N-1,
\end{equation}
where $\langle f,g\rangle_{L^2}=\int \Bar{f}g\d x$ represents the standard $L^2$ inner product, and $\|f\|_{L^2} = \langle f,f\rangle_{L^2}$.
In~\cref{fig:trig-poly}, we demonstrate the real and imaginary parts of the function $\psi_0(x)$ with $N = 32$. Intuitively, the functions $\psi_k$ are continuous interpolations of the ``square root'' of delta functions, i.e., $\psi_k \approx \sqrt{\delta_{x_k}}$. 

\begin{figure}
    \centering
    \includegraphics[width=0.5\linewidth]{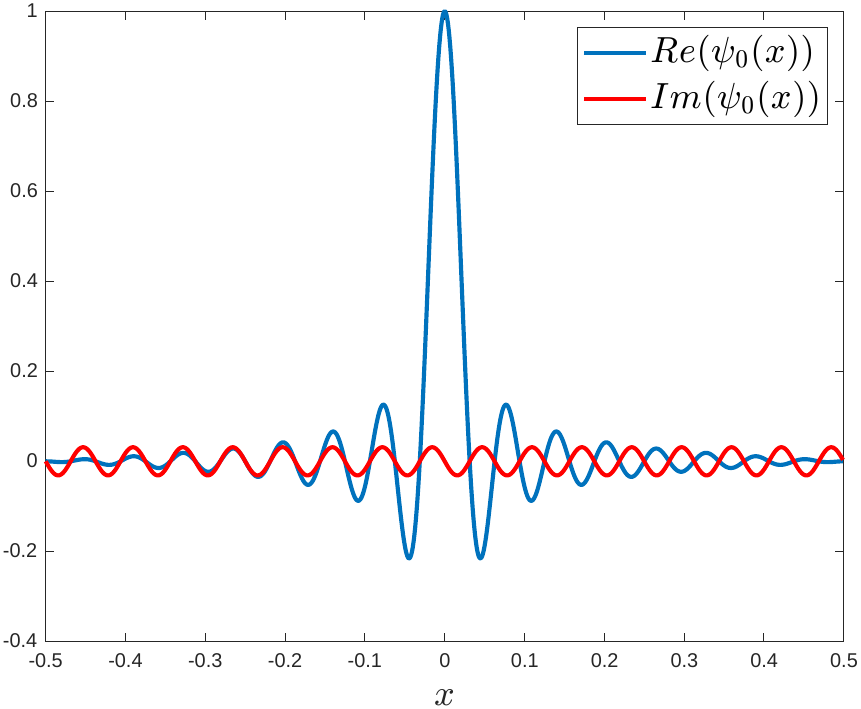}
    \caption{Real and imaginary parts of the trigonometric polynomial $\psi_0(x)$ (normalized such that $\psi_0(0) = 1$) over $[-1/2, 1/2]$ with $N=32$.}
    \label{fig:trig-poly}
\end{figure}

For a quantum state $\ket{u} = \sum^{N-1}_{k_1,\dots,k_d = 0} u(k_1,\dots,k_d)\ket{k_1,\dots,k_d}$, we define the following interpolation map, which can be regarded as an isometry that embeds $\CC^{N^d}$ into $L^2(\Omega)$:
\begin{equation}\label{eqn:interp-map}
    I_N\ket{u} = \sum^{N-1}_{k_1,\dots,k_d = 0} u(k_1,\dots,k_d)\psi_{k_1}(x_1)\dots \psi_{k_d}(x_d).
\end{equation}

Let $\ket{\widetilde{\sqrt{\sigma}}}$ be the right singular vector of $\widetilde{\mathbb{L}}$ associated with the smallest singular value (or equivalently, it is the ground state of the discretized Witten Laplacian $\widetilde{H} = \widetilde{\mathbb{L}}^\dagger \widetilde{\mathbb{L}}$). 
Through the interpolation~\cref{eqn:interp-map}, we can quantify the discretization error by analyzing $\left\|I_N\ket{\widetilde{\sqrt{\sigma}}} - \ket{\sqrt{\sigma}}\right\|_{L^2}$. This will be made rigorous in~\cref{assump:discretization}.

Moreover, the interpolation also provides a practical recipe for achieving high-accuracy Gibbs sampling. 
If we directly measure $\ket{\widetilde{\sqrt{\sigma}}}$ using the computational basis, we would realize a random variable $\tilde{X}$ whose distribution is completely supported on the grid points of $\mathcal{M}$. Since the law of $\tilde{X}$ is a finite mixture of Dirac measures, the TV distance between $\tilde{X}$ and the true Gibbs measure $\sigma$ is always $1$!
A naive fix could be to spread the probability mass at each mesh point to its neighborhood (like what we did in the classical post-processing in~\cref{lem:high-accuracy-interpolation}). In this case, the TV distance decays at a rate $\mathcal{O}(1/N)$ because the law of the new random variable is approximately a ``Riemann sum'' approximation of $\sigma$. This decay rate is unfavorable as it implies that the discretization number has to be $N = \mathcal{O}(1/\epsilon)$, which would inevitably lead to a sampling algorithm with a query complexity $\poly(1/\epsilon)$ (i.e., a low-accuracy sampler). 
In~\cref{lem:high-accuracy-interpolation}, we exploit the trigonometric interpolation to achieve high-accuracy sampling with a potentially small discretization number $N$.

\begin{lem}[Resolution booster]\label{lem:high-accuracy-interpolation}
    For a fixed $\epsilon > 0$, suppose that we have access to a state $\ket{g}\in \CC^{N^d}$ such that 
    \begin{equation}\label{eqn:assump-interpolation}
        \left\|I_N\ket{g} - \sqrt{\sigma}\right\|_{L^2} \le \epsilon/2,
    \end{equation}
    where $\ket{\sqrt{\sigma}} \in L^2(\RR^d)$ is the encoded Gibbs state.\footnote{Note that the interpolation function $I_N\ket{g}$ is a function defined in $\Omega \subset \RR^d$. It can be naturally extended to the whole space with the same $L^2$-norm.}
    Then, there is a quantum algorithm that outputs a random variable $X$ following the distribution $\eta$ such that $\rm{TV}(\eta, \sigma) \le \epsilon$ with $1$ copy of the state $\ket{g}$, and an additional $d\cdot \polylog(1/\epsilon)$ elementary gates.
\end{lem}
\begin{proof}
    Without loss of generality, we assume the $\log_2(N) = q$ is a positive integer, i.e., $\ket{g}$ can be represented by $dq$ qubits. For a given integer $r \ge q$ (the choice of $r$ will be specified later), we consider the following unitary operator:
    \begin{equation}
        U\colon \ket{0^{r-q}}\otimes \ket{k} \mapsto \ket{\overline{\psi_k}},\quad \forall k = 0,\dots,2^q-1,
    \end{equation}
    where  
    \begin{equation}
        \ket{\overline{\psi_k}} = \frac{1}{\mathcal{N}_k} \sum^{2^r-1}_{j=0} \psi_k(j/2^r) \ket{j},\quad \mathcal{N}_k = \left(\sum^{2^r-1}_{j=0} |\psi_k(j/2^r)|^2\right)^{1/2}.
    \end{equation}
    This operation can be constructed by performing a discrete Fourier transform on $q$ qubits, ``padding'' the resulting state with a state $\ket{0^{r-q}}$, and performing an inverse discrete Fourier transform on the $q+(r-q)=r$ qubits.     Applying this argument to every dimension, we can add $d(r-q)$ ancilla qubits to the original register that stores $\ket{g}$ and apply $U^{\otimes d}$ to the extended register. The resulting quantum state reads:
    \begin{equation}
        U^{\otimes d}\ket{0^{d(r-q)}}\otimes \ket{g} = \sum^{N-1}_{k_1,\dots,k_d = 0} g(k_1,\dots,k_d)\ket{\overline{\psi_{k_1}}}\dots \ket{\overline{\psi_{k_d}}},
    \end{equation}
    which can be regarded as a spatial discretization of $I_N\ket{g}$ with discretization number $M = 2^r$.

    Now, by measuring the quantum state $U^{\otimes d}\ket{g}$ using the computational basis, we will obtain a random variable with $M^d$ possible outcomes. These outcomes can be identified with a finer regular mesh in $\Omega$ with $M^d$ quadrature points. Then, for an outcome $\mathbf{z} = (z_1,\dots,z_d)$ (where $z_k \in \{0,1/2^r,\dots, 1-1/2^r\}$ for $k \in [d]$), we uniformly sample a point in the box centered at $\mathbf{z}$ with an edge length $1/M$. Following this protocol, we realize a random variable $X \in \Omega$ whose probability density $\eta(x)$ is a piece-wise constant function. In particular, this $\eta$ can be regarded as a ``Riemann sum'' approximation of the function $\tilde{\sigma} \coloneqq |I_N\ket{u}|^2$. Therefore, by choosing $M = \poly(1/\epsilon)$, we can have
    \begin{equation}\label{eqn:final-1}
        \int_\Omega |\eta(x) - \tilde{\sigma}(x)|~\d x \le \epsilon.
    \end{equation}
    Moreover, by applying Cauchy-Schwarz inequality to~\cref{eqn:assump-interpolation}, we have
    \begin{equation}\label{eqn:final-2}
        \int |\tilde{\sigma}(x) - \sigma(x)|~\d x \le \left\|I_N\ket{u}+\sqrt{\sigma}\right\|_{L^2} \left\|I_N\ket{u}-\sqrt{\sigma}\right\|_{L^2} \le \epsilon.
    \end{equation}
    Combining~\cref{eqn:final-1} and~\cref{eqn:final-2} using the triangle inequality, we obtain the desired estimate on the TV distance:
    \begin{equation}
        \mathrm{TV}(\eta, \sigma) = \frac{1}{2}\int |\eta(x) - \sigma(x)|~\d x \le \frac{\epsilon}{2} + \frac{\epsilon}{2} = \epsilon.
    \end{equation}
    Note that gate complexity of $U^{\otimes d}$ is $\mathcal{O}(d\poly(r)) = d\cdot \polylog(1/\epsilon)$, where $r = \log_2(M) = \log_2(\poly(1/\epsilon))$.
\end{proof}

\begin{rem}
    A similar interpolation technique has been introduced in~\cite[Theorem 3.2]{motamedigibbs} to exponentially improve the sampling resolution. To ensure the Gibbs measure can be efficiently interpolated using a slightly different set of trigonometric polynomials,~\cite{motamedigibbs} requires the target Gibbs measure $\sigma$ to be ``semi-analytical'', and the purpose is similar to our~\cref{assump:discretization}, as detailed below.
\end{rem}

For a smooth $V$ with a sufficient growth rate, the state $\sqrt{\sigma}$ is also smooth and has a fast-decaying tail. 
In this case, the DFT-based discretization (also known as a pseudo-spectral method in numerical analysis) often exhibits \textit{spectral convergence}, i.e., the discretization error decays super-polynomially fast with an increasing discretization number $N$.
In this paper, we make the following assumption that ensure the efficiency of the spatial discretization of Langevin dynamics:

\begin{assump}[Spatial discretization of Langevin dynamics]\label{assump:discretization}
    Let $V\colon \RR^d \to \RR$ be a smooth potential function, and $\sigma \propto e^{-\beta V}$ be the Gibbs measure. 
    For an arbitrary $\epsilon>0$, we assume that we can choose $a = \mathcal{O}(\log(d/\epsilon))$ and $N = a\cdot\poly\log(d/\epsilon)$ such that the followings hold:
    \begin{enumerate}
        \item The discretized Witten Laplacian $\widetilde{\mathcal{H}} \coloneqq \sum^d_{j=1} \tildeL^\dagger_j\tildeL_j$ has a ground state $\ket{\widetilde{\sqrt{\sigma}}}$ that satisfies $\left\|I_N \ket{\widetilde{\sqrt{\sigma}}} - \sqrt{\sigma}\right\|_{L^2} \le \epsilon/4$, 
        \item Compared to the Witten Laplacian $\mathcal{H}$ in~\cref{eqn:spatial_discretize_Lj}, the smallest eigenvalue of $\widetilde{\mathcal{H}}$ is no greater than $\mathrm{Gap}\left(\mathcal{H}\right)/16$, and the second smallest eigenvalue of $\widetilde{\mathcal{H}}$ is no smaller than $9 \mathrm{Gap}\left(\mathcal{H}\right)/ 16$. 
    \end{enumerate}
\end{assump}

We numerically implemented our DFT-based discretization scheme, and the results suggest that such a discretization scheme achieves an accuracy $\epsilon$ with only poly-logarithmically large truncation length $a$ and discretization number $N$ for several non-convex potentials.

\subsection{Block-encoding of $\mathbb{L}$}\label{append:pseudo-diff-discretize}

In this subsection, we discuss how to efficiently block-encode the operator $\mathbb{L}$ using a quantum computer. 
We assume access to the gradient (i.e., first-order) oracle of $V$:
\begin{align*}
    O_{\nabla V}\colon \ket{\xx}\ket{\zz}\mapsto \ket{\xx}\ket{\nabla V(\xx) + \zz} = \ket{\xx}\ket{\partial_{1} V(\xx)+z_1,\dots, \partial_d V(\xx)+z_d}
\end{align*}
and its inverse $O^\dagger_{\nabla V}$.
Each sub-register (i.e, $\ket{\xx}$ and $\ket{\zz}$, respectively) consists of $bd$ qubits, where $b$ is the number of bits for the fixed point number representation and $d$ is the number of components in the vector. 
For simplicity, we assume $b$ is a fixed constant (e.g., $32$), and the problem dimension $d$ is a \textit{power of $2$}, i.e., $\log_2(d)$ is a positive integer.

The block-encoding of each $\tildeL_j$ is based on the spatial discretization~\cref{eqn:dft-discretized-Lj}, which requires us to block-encode the following matrices:
\begin{equation}\label{eqn:P_Q_j}
    \tilde{P}_j = \frac{2\pi}{\sqrt{\beta}}\left(U^\dagger_{\rm FT}\right)^{\otimes d} \left(\Phi^j_{{\bm \xi}}\right) U_{\rm FT}^{\otimes d}, \quad \tilde{Q}_j =  - i\frac{\sqrt{\beta}}{2} \Phi^j_{\xx}.
\end{equation}

\begin{lem}[Select oracle]\label{lem:select-oracle}
Let 
$$\mathrm{SEL} = \sum^d_{j=1}\ketbra{j}{j}\otimes U_{\tilde{P}_j},$$
where $U_{\tilde{P}_j}$ is a $(\pi N/\sqrt{\beta}, 1, 0)$-block-encoding of $\tilde{P}_j$. We can implement the unitary $\mathrm{SEL}$ with $d\cdot \polylog(d, N)$ elementary gates.
\end{lem}
\begin{proof}
    The block-encoding of the select oracle can be implemented as follows.
    We consider the following register:
    \begin{align}
        \underbrace{\ket{000\dots 0}}_{\text{dim. index}}\otimes \underbrace{\ket{{\bm \xi}}}_{\text{frequency number}} \otimes \underbrace{\ket{0}}_{\text{ancilla}}.
    \end{align}
    The first sub-register (for dimension indexing) consists of $ \log_2(d)$ qubits\footnote{Here, we assume $d$ is a power of $2$ so all computational basis in the indexing register is used.}, the second sub-register (for frequency number) consists of $db$ qubits, and the last (ancilla) sub-register has $1$ qubit.
    
    To compute the diagonal differential coefficient matrix $\Phi^j_{\bm \xi}$, we can implement a sequence of $b$ controlled rotations, each controlled on the index sub-register, and perform a rotation on the ancilla qubit. These controlled rotations use classical arithmetic operators\footnote{Any classical arithmetic operations expressed as logical circuits can be implemented on a quantum computer using the same number of elementary gates, see~\cite[Chapter 4]{nielsen2010quantum}.} to compute the normalized frequency $2\xi_j/N \in [-1,1]$, and can be executed by $\poly(b)\cdot\log(d)$ elementary gates, where $\poly(b)$ is the cost for implementing classical arithmetic operators, and $\log(d)$ is the overhead for the controlled rotation on $\ket{j}$ basis.
    We obtain the following state:
    \begin{equation}
        \ket{j}\otimes \ket{{\bm \xi}}\otimes \ket{0} \mapsto \ket{j}\otimes \ket{{\bm \xi}}\otimes \left(\frac{\xi_j}{N/2}\ket{0} + \sqrt{1 - \frac{ \xi^2_j}{(N/2)^2}}\ket{1}\right).
    \end{equation}
    In other words, this implements a $(N/2,1)$-block-encoding of $\ketbra{j}{j}\otimes \tilde{\Phi}^j_{\bm \xi_j}$. By sandwiching this unitary operator between the $d$-dimensional QFT and its inverse, we obtain a controlled unitary $U'_{\tilde{P}_j} = \ketbra{j}{j}\otimes U_{\tilde{P}_j} + \sum_{k\neq j}\ketbra{k}{k}\otimes I$, where $U_{\tilde{P}_j}$ is a $(\pi N/\sqrt{\beta}, 1)$-block-encoding of $\tilde{P}_j$. 
    Note that the $d$-dimensional QFT (or its inverse) can be implemented using $\mathcal{O}(d\log^2(N))$ elementary qubits.
    By concatenating $d$ such controlled unitaries in a sequence, we can implement the select oracle $\mathrm{SEL}$ with $1$ ancilla qubit and $d\cdot \poly\log(d, N)$ elementary gates.
\end{proof}

Let $\mathcal{A} \in \CC^{dN^d\times dN^d}$ be the following block matrix: 
\begin{equation}\label{eqn:cal-A}
    \mathcal{A} = [\tilde{L}^\top_1, \tilde{L}^\top_2,\dots, \tilde{L}^\top_d]^\top
\end{equation}
We can write $\mathcal{A} = \mathcal{A}_1 + \mathcal{A}_2$ with ($\tilde{P}_j$, $\tilde{Q}_j$ are the same as in~\cref{eqn:P_Q_j})
\begin{equation*}
    \mathcal{A}_1 = [\tilde{P}^\top_1, \tilde{P}^\top_2,\dots, \tilde{P}^\top_d]^\top,\quad \mathcal{A}_2 = [\tilde{Q}^\top_1, \tilde{Q}^\top_2,\dots, \tilde{Q}^\top_d]^\top.
\end{equation*}

\begin{lem}\label{lem:A1-block-encode}
    We can implement a $(\pi N\sqrt{d/\beta}, 1)$-block-encoding of the matrix $\mathcal{A}_1$ with $d\cdot \polylog(d,N)$ elementary gates.
\end{lem}

\begin{proof}
Recall that we assume the problem dimension $d$ is a power of $2$. Therefore, we can prepare a uniform superposition state $\ket{\Psi} = \frac{1}{\sqrt{d}}\sum^{d-1}_{j=0}\ket{j}$ using $\log_2(d)$ Hadamard gates.
Then, by applying the layer of Hadamard gates followed by the select oracle $\mathrm{SEL}$, we prepared a $(\pi N\sqrt{d/\beta}, 1)$-block-encoding of $\mathcal{A}_1$. The total number of elementary gates is $d\cdot \polylog(d,N)$.
\end{proof}

\begin{rem}
    Note that the $\mathcal{O}(\sqrt{d})$ overhead in the normalization factor of the block-encoding of $\mathcal{A}_1$ is inevitable because $\|\mathcal{A}_1\| = \Omega(\sqrt{d})$.
\end{rem}

\begin{defn}
    A differentiable function $V \colon \RR^d \to \RR$ is $\ell$-smooth if for any $x, y \in \RR^d$, 
    \begin{equation*}
        \|\nabla V(x) - \nabla V(y)\| \le \ell \|x -y\|,
    \end{equation*}
    where $\|\cdot\|$ stands for the standard Euclidean distance in $\RR^d$. In other words, $V$ is $\ell$-smooth if its gradient is $\ell$-Lipschitz continuous.
\end{defn}

In the following lemma, we show that a $(\sqrt{\beta}R/2, 1)$-block-encoding of $\mathcal{A}_2$ can be constructed using 2 queries to $O_{\nabla V}$. Here, $R = \max_{\xx \in \Omega} \|\nabla V(\xx)\|$.
For an $\ell$-smooth potential $V$ with at least 1 local minimum in the numerical domain $\Omega$, it is clear that $R \le a\ell \sqrt{d} = \mathcal{O}(\ell \sqrt{d}\log(d/\epsilon))$ because $\Omega$ has a diameter $a\sqrt{d}= \mathcal{O}(\sqrt{d}\log(d/\epsilon))$.

\begin{lem}\label{lem:A2-block-encode}
    Let $R = \max_{\xx \in \Omega} \|\nabla V(\xx)\|$. 
    We can implement a $(\sqrt{\beta}R/2, 1)$-block-encoding of $\mathcal{A}_2$ with 2 queries to the gradient oracle $O_{\nabla V}$ or its inverse, and an additional $\tilde{\mathcal{O}}(d^2)$ elementary gates.
\end{lem}
\begin{proof}
    We introduce an indexing register with $d$ qubits (i.e., ``dim. index'') to achieve this goal. Consider the following register:
    \begin{align}
        \underbrace{\ket{0^{\otimes m_1}}}_{\text{dim. index}}\otimes \underbrace{\ket{\xx}}_{\text{mesh point}} \otimes \underbrace{\ket{0^{\otimes bd}}}_{\text{anc. grad.}} \otimes \underbrace{\ket{0}}_{\text{ancilla}}.
    \end{align}
    The sizes of the sub-registers are: $m_1 = \lceil \log_2(d)\rceil $, $d\log_2(N)$, $bd$ (recall that $b$ is fixed bit-width precision), and $1$. 
    First, we apply the oracle $O_{\nabla V}$ on the second and the third sub-register, and the state becomes: 
    \begin{align}\label{eqn:gradient-query-1}
        \ket{0^{\otimes m_1}}\otimes \ket{\xx}\otimes \ket{\partial_1 V(\xx), \partial_2 V(\xx),\dots,\partial_d V(\xx)} \otimes \ket{0}.
    \end{align}
    Next, we consider a sequence of $d$ ``controlled'' rotation gadgets, while using the first register as a counting register. 
    The first rotation is controlled on the third register and rotates the basis in the (combined) first and last sub-register:
    \begin{equation}\label{eq:first-givens-rotation}
        \ket{0^{m_1}}\otimes \ket{0} \mapsto \left(\frac{\partial_1 V(\xx)}{R}\ket{0}\ket{0} + \frac{\sqrt{R^2-|\partial_1V(\xx)|^2}}{R} \ket{0}\ket{1}\right).
    \end{equation}
    Note that the $\ket{1}\ket{0}$ state in the right-hand side of~\cref{eq:first-givens-rotation} represents the tensor product of $\ket{1}$ in the first sub-register and $\ket{0}$ in the last sub-register.
    The interpretation of $\ket{0}\ket{1}$ and other product states that follow is similar.
    In other words, this is a controlled Givens rotation operation (i.e., a rotation acting on 2 entries in the computational basis) that encodes the first partial derivative $\partial_1 V(\xx)$ to the basis $\ket{0}\ket{0}$ and the remainder in the basis $\ket{0}\ket{1}$. This rotation gadget can be efficiently implemented by controlling on the first $b$ qubits in the third sub-register.

    The second rotation is again controlled on the gradient sub-register and acts non-trivially on two basis states:
    \begin{equation}
        \ket{0}\ket{1} \mapsto \sin(\theta) \ket{0}\ket{1} + \cos(\theta) \ket{1}\ket{0}, 
    \end{equation}
    where
    \begin{equation}
        \sin(\theta) = \frac{\sqrt{R^2-|\partial_1V(\xx)|^2 - |\partial_2V(\xx)|^2}}{\sqrt{R^2-|\partial_1V(\xx)|^2}}, \quad \cos(\theta) = \frac{\partial_2 V(\xx)}{\sqrt{R^2-|\partial_1V(\xx)|^2}}.
    \end{equation}
    By applying this rotation to the state~\cref{eq:first-givens-rotation}, we encode the second partial derivative $\partial_2 V$ to the amplitude of $\ket{1}\ket{0}$:
    \begin{equation}
        \left(\frac{\partial_1 V(\xx)}{R}\ket{0}\ket{0} + \frac{\partial_2 V(\xx)}{R}\ket{1}\ket{0} +\frac{\sqrt{R^2-|\partial_1V(\xx)|^2 - |\partial_2V(\xx)|^2}}{R} \ket{0}\ket{1}\right).
    \end{equation}
    It is worth noting that the rotation angle has to be carefully computed (via classical arithmetic circuits) to keep track of the previous subnormalization factors, which requires reading the first $2b$ qubits in the third sub-register. 
    
    Iterating this process for all $\ket{j}\ket{0}$ basis for $j = 0, \dots, d-1$, we will prepare a quantum state that encodes the $j$-th partial derivative of $V$ in the amplitude of the computational basis $\ket{j-1}\ket{0}$. Since each Givens rotation has to compute the subnormalization factor using all previous gradients, the overall number of operations scales as $1 + 2+ \dots + d = \mathcal{O}(d^2)$.
    Finally, we apply the inverse of $O_{\nabla V}$ to uncompute the third register (i.e., $\ket{\nabla V(\xx)}$) and discard it thereafter.
    The resulting state is
    \begin{align}\label{eqn:gradient-query-3}
        \sum^d_{j=1}\frac{\partial_j V(\xx)}{R}\ket{j}\ket{\xx}\ket{0} + \ket{\perp},
    \end{align}
    which implements a $(\sqrt{\beta}R/2, 1)$-block-encoding of $\mathcal{A}_2$. 
\end{proof}

Finally, we can perform LCU to construct a block-encoding of $\mathcal{A}$ (the discretized $\mathbb{L}$).

\begin{prop}[Block-encoding of $\mathbb{L}$]\label{prop:block-encode-A} 
    Let $N$ and $R$ are the same as above, and $\alpha=\pi N\sqrt{d/\beta}+\sqrt{\beta}R/2$.
    We can implement an $(\alpha, 3)$-block-encoding of the matrix $\mathcal{A}$ with 2 queries to the gradient oracle $O_{\nabla V}$ (or its inverse) and an additional $\widetilde{\mathcal{O}}(d^2)$ elementary gates. Here, the $\widetilde{\mathcal{O}}(\cdot)$ notation suppresses poly-logarithmic factors in $d$ and $N$.
\end{prop}
\begin{proof}
    Note that $\mathcal{A} = \mathcal{A}_1 +\mathcal{A}_2$, and the block-encoding of $\mathcal{A}_1$ and $\mathcal{A}_2$ can be constructed using Lemma~\ref{lem:A1-block-encode} and Lemma~\ref{lem:A2-block-encode}, respectively.
    Effectively, they can be regarded as block-encodings of $\mathcal{A}_1/\alpha_1$ and $\mathcal{A}_2/\alpha_2$, each with subnormalization factors $1$.
    Now, we invoke the Linear Combination of Unitaries (LCU) technique~\cite[Lemma 29]{gilyen2019quantum} to construct a block-encoding of $\mathcal{A}$. Due to the different subnormalization factors, we use an LCU coefficient pair $y = (\alpha_1/(\alpha_1+\alpha_2), \alpha_2/(\alpha_1+\alpha_2)$. As a result, we obtain a $(\|y\|_1, 3)$-block-encoding of
    \begin{equation}
        y_1 \left(\frac{\mathcal{A}_1}{\alpha_1}\right) + y_2 \left(\frac{\mathcal{A}_2}{\alpha_2}\right) = \frac{\mathcal{A}}{(\alpha_1+\alpha_2)},
    \end{equation}
    which corresponds to a block-encoding of $\mathcal{A}$ with a normalization factor
    \begin{equation}
        \alpha = \pi N \sqrt{d/\beta} + \sqrt{\beta}R/2.
    \end{equation}
    Note that the number of ancilla qubits in the block-encoding of $\mathcal{A}$ is $3$ since $\mathcal{A}_1$ and $\mathcal{A}_2$ each require $1$ ancilla qubit, and LCU requires an additional ancilla qubit.
\end{proof}

\subsection{Proof of Theorem~\ref{thm:main}}\label{append:proof-main}

The following result is a rigorous version of~\cref{thm:main}. 

\begin{thm}\label{thm:main-formal}
    Suppose that~\cref{assump:discretization} holds and the potential $V$ is $\ell$-smooth. 
    Assume access to a state $\ket{\phi}$ (i.e., warm start) such that $\abs{\braket{\phi}{\widetilde{\sqrt{\sigma}}}} = \Omega(1)$.
    Then, there exists a quantum algorithm that outputs a random variable $X\in \RR^d$ distributed following a probability distribution $\eta$ such that $\mathrm{TV}(\eta, \sigma) \le \epsilon$ using 
    \begin{equation}\label{eqn:main-1-query-complexity}
        \sqrt{\beta d C_{\rm PI}} \cdot \left(\beta^{-1}+\ell\right) \cdot \polylog(d, \epsilon^{-1})
    \end{equation}
    quantum queries to the gradient $\nabla V$, and $\mathcal{O}(1)$ copies of the state $\ket{\phi}$.
\end{thm}

\begin{proof}
    For $j = 1,\dots,d$, we denote $\tildeL_j$ as the spatially discretized $L_j$ as given in~\cref{eqn:dft-discretized-Lj}.
    We write $\widetilde{\mathbb{L}} \coloneqq [\tildeL^\top_1, \tildeL^\top_2,\dots, \tildeL^\top_d]^\top$.
    By Definition~\ref{defn:be-general}, the block-encoding of $\mathcal{A}$ constructed in Proposition~\ref{prop:block-encode-A} provides a $(\alpha, 3)$-block-encoding of $\widetilde{\mathbb{L}}$ with the normalization factor $\alpha = \pi N\sqrt{d/\beta}+\sqrt{\beta}R/2$. 
    For a fixed accuracy parameter $\epsilon > 0$, due to~\cref{assump:discretization}, we can truncate the numerical domain to $\Omega = [-a,a]^d$ with an edge length $a=\mathcal{O}(\log(d/\epsilon))$, which implies that the diameter of $\Omega$ is $\mathcal{O}(\sqrt{d}\log(d/\epsilon))$. Then, by the $\ell$-smoothness of $V$, we have $R = \max_{x\in \Omega} \|\nabla V(x)\| = \mathcal{O}(\ell \sqrt{d}\log(d/\epsilon))$.
    It follows that
    \begin{equation}
        \alpha = \pi N\sqrt{d/\beta}+ \sqrt{\beta}R/2 = \mathcal{O}\left(N\sqrt{\beta d}\cdot \log(d/\epsilon) \cdot (\beta^{-1}+ \ell)\right).
    \end{equation}
    Recall that the discretized Witten Laplacian $\widetilde{\mathcal{H}} = \sum^d_{j=1}\tildeL^\dagger_j\tildeL_j$ has a unique ground state $\ket{\widetilde{\sqrt{\sigma}}}$.
    Given access to a warm start state $\ket{\phi}$ such that $\abs{\braket{\phi}{\widetilde{\sqrt{\sigma}}}} = \Omega(1)$, by~\cref{thm:main-qsvt}, we can prepare a quantum state $\ket{g}$ such that $\left\|g - \widetilde{\sqrt{\sigma}}\right\|\le \epsilon/4$ with
    \begin{equation}\label{eqn:main-query-pre}
        \mathcal{O}\left(\frac{\alpha}{\sqrt{\mathrm{Gap}\left(\mathcal{H}\right)}}\right) =  \sqrt{\beta d C_{\rm PI}} \cdot (\beta^{-1}+ \ell) \cdot \poly\log(d, \epsilon^{-1})
    \end{equation}
    queries to the block-encoding of $\mathcal{A}$, which amounts to the same query complexity to $O_{\nabla V}$ due to~\cref{prop:block-encode-A}. Since the success probability is $\Omega(1)$ due to warm start, the total number of copies of $\ket{\phi}$ required is $ \mathcal{O}(1)$.
    By~\cref{assump:discretization}, $\left\|I_N \ket{\widetilde{\sqrt{\sigma}}} - \sqrt{\sigma}\right\|_{L^2}\le\epsilon/4$ and $I_N$ is an isometry.
    Thus, due to the triangle inequality,
    \begin{equation}
        \left\|I_N\ket{g} - \sqrt{\sigma}\right\|_{L^2} \le \epsilon/2.
    \end{equation}
    Now, by invoking~\cref{lem:high-accuracy-interpolation}, we can realize a random variable $X \sim \eta$ such that $\mathrm{TV}(\eta, \sigma) \le \epsilon$ with the number of queries to $\nabla V$ given by~\cref{eqn:main-1-query-complexity}.
\end{proof}

\begin{rem}
For logconcave sampling, our query complexity exhibits scaling invariance with respect to the temperature $T = 1/\beta$. In particular, when the potential $V$ is $\gamma$-strongly convex, the Poincar\'e constant always reads $C_{\rm PI} = 1/\gamma$. Then, the query complexity of our quantum algorithm is 
\begin{equation}
    \sqrt{\frac{d}{\beta \gamma}} \cdot (1+ \ell \beta) \cdot \polylog(d, \epsilon^{-1}).
\end{equation}
\end{rem}

\begin{rem}\label{rem:warm-start} For the warm start condition, we have two comments:
\begin{itemize}
\item Our warm start condition $\abs{\braket{\sqrt{\rho_0}}{\sqrt{\sigma}}}=\Omega(1)$ does not require that $\rho_0$ has mass on every region where $\sigma$ does. For example, let $\sigma(x)$ be a mixture of two Gaussians: $\sigma(x)\propto \exp(-|x|^2/(2R^2))+\exp(-|x-a|^2/(2R^2))$, and $\rho(x)\propto \exp(-|x|^2/(2R^2))$. When $R\ll 1$, we have $\abs{\braket{\sqrt{\sigma}}{\sqrt{\rho_0}}} \approx 1/2=\Omega(1)$. However, it is straightforward to see that $\sigma(B_a(R))=\Omega(1)$ and $\rho_0(B_a(R))\ll 1$, showing that $\rho_0$ does not ``cover'' the entire support of $\sigma$.

\item We note that the warm-start assumption in~\cite{childs2022quantum} is strictly stronger than ours. 
In~\cite[(C.15)]{childs2022quantum}, it is assumed that the initial state $\ket{\sqrt{\rho_0}}\in L^2(\RR^d)$ satisfying $\|\rho_0\|_{L^2}=1$ is $\beta$-\textit{warm}, i.e.,
\begin{equation}\label{eqn:warm_start_stronger}
    \sup_{x} \rho_0(x)/\sigma(x) \le \beta,
\end{equation}
where $\beta > 0$ is a constant independent of $C_{\rm PI}$ and $d$.
This condition immediately implies our warm start assumption: for a $\beta$-warm initial state $\ket{\psi}$, we have
\begin{equation}
    \braket{\sqrt{\rho_0}}{\sqrt{\sigma}} = \int_{\RR^d} \sqrt{\rho_0(x)\sigma(x)}~\d x \ge \beta\int_{\RR^d} \rho_0~\d x = \beta = \Omega(1).
\end{equation}
Furthermore, we note that it suffices to assume
\begin{equation}\label{eqn:warm_start_weaker}
\min\left\{\sup_{x} \frac{\rho_0(x)}{\sigma(x)}, \sup_{x} \frac{\sigma(x)}{\rho_0(x)}\right\} \le \beta,
\end{equation}
in order to guarantee that $\braket{\sqrt{\rho_0}}{\sqrt{\sigma}}$ is at least $\beta$. As a concrete example, let $\rho_0(x)$ be a mixture of two Gaussians: $\rho_0(x)\propto \exp(-|x|^2/(2R^2))+\exp(-|x-a|^2/(2R^2))$, and let $\rho_1(x)\propto \exp(-|x|^2/(2R^2))$. When $R\ll 1$, we have $\abs{\braket{\sqrt{\rho_0}}{\sqrt{\rho_1}}} \approx 1/2$, even though $\rho_0$ is not a constant-warm start initial state for $\rho_1$ according to~\cref{eqn:warm_start_stronger}, as the warmness parameter $\beta = \Omega\left(\exp(|a|^2/R^2)/R^d\right)$ becomes exponentially large as $R\rightarrow 0$.
\end{itemize}
\end{rem}

\section{Operator-Level Acceleration of Replica Exchange: Details}\label{append:re-complexity-analysis}
\subsection{Derivation of the generalized Witten Laplacian}\label{append:symmetrize-reld}

In this section, we derive the generalized Witten Laplacian for replica exchange Langevin diffusion (RELD). Recall that, for two inverse temperatures $\beta > \beta' > 0$, the forward Kolmogorov equation of RELD is given by
\begin{equation}\label{eqn:full-generator-reld}
    \begin{aligned}
    \partial_t \rho = \mathcal{L}(\rho) \coloneqq &\underbrace{\nabla_x \cdot (\nabla_x V(x) \rho(t,x,y)) + \beta^{-1} \Delta_x \rho(t,x,y)}_{\mathcal{L}_1(\rho(t,x,y))} + \underbrace{\nabla_y \cdot (\nabla_y V(x) \rho(t,x,y)) + \beta'^{-1} \Delta_y \rho(t,x,y)}_{\mathcal{L}_2(\rho(t,x,y)} \\
    &+ \underbrace{\mu \left(s(y,x)\rho(t,y,x) - s(x,y)\rho(t,x,y)\right)}_{\mathcal{L}_s(\rho(t,x,y))}.
\end{aligned}
\end{equation}
The first two operators $\mathcal{L}_1$ and $\mathcal{L}_2$ correspond to the low- and high-temperature Langevin dynamics, respectively. The third operator $\mathcal{L}_s$ corresponds to a Metropolis-Hasting type swapping operations between the two continuous-time Markov chains, with the swapping probability given by
\begin{equation}
    s(x,y) = \min\left(1, \frac{\bsigma(y,x)}{\bsigma(x,y)}\right).
\end{equation}
For more details on the generator of the swapping mechanism ($\mathcal{L}_s$), we refer the readers to~\cite{DongTong2022}.

The generalized Witten Laplacian of RELD is obtained via the similarity transformation 
$$\mathcal{H} = - \bsigma^{-1/2}\circ \mathcal{L} \circ \bsigma^{1/2},$$ where the joint Gibbs measure is the invariant measure of RELD dynamics:
\begin{equation*}
    \bsigma(x,y) \propto \exp\left(- \beta V(x) - \beta'V(y)\right).
\end{equation*}
We can compute the generalized Witten Laplacian for each component operator in $\mathcal{L}$. When applying the similarity transformation to the operator $\mathcal{L}_1$, we obtain the Witten Laplacian corresponding to the low-temperature Langevin dynamics:
\begin{align}
    \mathcal{H}_1 = - \bsigma^{-1/2}\circ \mathcal{L}_1 \circ \bsigma^{1/2} = - \frac{1}{\beta} \Delta_x + \left(\frac{\beta}{4}\|\nabla_x V\|^2 - \frac{1}{2}\Delta_x V\right).
\end{align}
The similarity transformation also maps $\mathcal{L}_2$ to the Witten Laplacian of the high-temperature Langevin dynamics:
\begin{align}
    \mathcal{H}_2 = - \bsigma^{-1/2}\circ \mathcal{L}_2 \circ \bsigma^{1/2} = - \frac{1}{\beta'} \Delta_y + \left(\frac{\beta'}{4}\|\nabla_y V\|^2 - \frac{1}{2}\Delta_y V\right).
\end{align}
These are frustration-free Hamiltonians and can be factorized in the same way as the standard Witten Laplacian:
\begin{align}
    \mathcal{H}_1 = \sum^d_{j=1} L^{\dagger}_j L_j,\quad L_j \coloneqq -i \frac{1}{\sqrt{\beta}} \partial_{x_j} - i\frac{\sqrt{\beta}}{2}\partial_{x_j} V,
\end{align}
\begin{align}
    \mathcal{H}_2 = \sum^d_{j=1} L^{',\dagger}_j L'_j,\quad L'_j \coloneqq -i \frac{1}{\sqrt{\beta'}} \partial_{y_j} - i\frac{\sqrt{\beta'}}{2}\partial_{y_j} V.
\end{align}
The similarity transformation of the swap operator $\mathcal{L}_s$ is more involved and has been discussed in the main text (see~\cref{eqn:L_s_similarity}). The corresponding self-adjoint operator reads:
\begin{align}
    \mathcal{H}_s \coloneqq - \bsigma^{-1/2}\circ \mathcal{L}_s \circ \bsigma^{1/2} = L^\dagger_s L_s,\quad L_s = \sqrt{\frac{\mu}{2}}(I - W) S^{1/2},
\end{align}
where $\mu$ is the swapping intensity, $S$ is the point-wise multiplication defined by $S[\psi(x,y)] \coloneqq s(x,y)\psi(x,y)$ with
\begin{align}\label{eqn:s-func-multiplicative}
    s(x, y) = \min\left(1, \frac{\bsigma(y,x)}{\bsigma(x,y)}\right) = \exp\left(0 \wedge (\beta - \beta')\cdot (V(x) - V(y))\right),
\end{align}
and $W[\psi(x,y)] = \psi(y,x)$ swaps the $x$ and $y$ variables in a test function $\psi$.
It can be readily verified that the encoded Gibbs state $\ket{\sqrt{\bsigma}}$ is annihilated by $L_s$:
\begin{equation}
    \begin{aligned}
    L_s \ket{\sqrt{\bsigma}} &= \sqrt{\frac{\mu}{2}}(I-W)S^{1/2} \sqrt{\bsigma(x,y)}\\
    &= \sqrt{\frac{\mu}{2}}(I-W)\left[\sqrt{\bsigma(x,y)}1_{A_1} + \sqrt{\bsigma(y,x)}1_{A_2}\right]\\
    &= \sqrt{\frac{\mu}{2}}\left[\sqrt{\bsigma(x,y)}1_{A_1} + \sqrt{\bsigma(y,x)}1_{A_2} - \sqrt{\bsigma(y,x)}1_{A_2} - \sqrt{\bsigma(x,y)}1_{A_1}\right]\\
    &= 0,
\end{aligned}
\end{equation}
where $A_1 = \{(x,y)\in \RR^{2d}\colon \bsigma(y,x) > \bsigma(x,y)\}$ and $A_2 = \RR^d \backslash A_1$. By symmetry, $W I_{A_1} = I_{A_2}$ and $W I_{A_2} = I_{A_1}$.
It follows that the generalized Witten Laplacian of RELD can be written as
\begin{align}
    \mathcal{H} = \mathbb{L}^\dagger_{\mathrm{RE}}\mathbb{L}_{\mathrm{RE}},
    \label{eq:H_scr_RE}
\end{align}
where $\mathbb{L}_{\mathrm{RE}}$ is a block matrix consisting of $(2d+1)$ operators:
\begin{align}\label{eq:L_bb_RE}
    \mathbb{L}_{\mathrm{RE}} = [L^\top_1,\dots, L^\top_d,{L'_1}^\top,\dots, {L'_d}^\top,L^\top_s]^\top.
\end{align}
Since the generalized Witten Laplacian $\mathcal{H}$ is frustration-free and the encoded Gibbs state $\ket{\sqrt{\bsigma}}$ is annihilated by all operators $L_j$, $L'_j$, and $L_s$. It is clear that $\ket{\sqrt{\bsigma}}$ is the ground state of $\mathcal{H}$ and is annihilated by the block matrix $\mathbb{L}_{\mathrm{RE}}$. Hence, we can employ the algorithmic recipe developed in Appendix~\ref{append:meta-algorithm-qsvt} to prepare the encoded Gibbs state $\ket{\sqrt{\bsigma}}$.

\subsection{Block-encoding of $\mathbb{L}_{\rm RE}$}

In this subsection, we discuss the cost of building a block-encoding of $\mathbb{L}_{\rm RE}$. 
The operator matrix $\mathbb{L}_{\mathrm{RE}}$ can be efficiently block-encoded using both the zeroth- and first-order information $V$. To compute the swapping rate $s(x,y)$, we require access to the function value (i.e., zeroth-order) oracle of $V$:
\begin{align*}
    O_{V}\colon \ket{\xx}\ket{z}\mapsto \ket{\xx}\ket{V(\xx)},
\end{align*}
and its inverse. Here, $\ket{V(\xx)}$ is a bit-string state that encodes a $b$-bit representation of $V(\xx)$.

Similar to the Langevin dynamics case, we first perform a spatial discretization of  $\mathbb{L}_{\rm RE}$. Note that each component operator $L_j$ in $\mathbb{L}_{\rm RE}$ acts on functions in $\RR^{2d}$. We first truncate the space $\RR^{2d}$ into a finite-sized domain $\Omega = [-a,a]^{2d}$ and discretize it using a uniform grid with $N^{2d}$ points. The discretized $L_j$ operator is a matrix of dimension $N^{2d}$. For a smooth potential with certain growth conditions, to achieve a target sampling accuracy $\epsilon$ in the TV distance, we can choose $a = \mathcal{O}(\log(d/\epsilon))$ and $N = a\cdot \polylog(d/\epsilon)$ We denote $R = \max_{x\in \Omega}\|\nabla V(x)\|$ as the sub-normalization factor of the gradient $\nabla V$ in the numerical domain $\Omega$.

\begin{prop}[Block-encoding of $\mathbb{L}_{\rm RE}$]\label{prop:block-encoding-L-RE}
    Let 
    \begin{equation}\label{eqn:normalization-L-RE}
        \alpha = \left(\pi^2 N^2 d(\beta^{-1}+\beta'^{-1}) + (\beta+\beta')R^2/4 + 2\pi N \sqrt{d}R + 2\mu\right)^{1/2},
    \end{equation}
    where $N$ and $R$ are the same as above.
    We can implement an $(\alpha,10)$-block-encoding of the (discretized) $\mathbb{L}_{\mathrm{RE}}$ using $2$ queries to the gradient oracle $O_{\nabla V}$ or its inverse, $4$ queries to the function value oracle $O_V$ or its inverse, and an additional $\widetilde{\mathcal{O}}(d^2)$ elementary gates.
\end{prop}
\begin{proof}
    For $j \in [d]$, the block-encoding of $L_j$ and $L'_j$ can be constructed using a similar method as described in~\cref{prop:block-encode-A}. It is worth noting that, while we have two replicas and each of them requires an independent gradient $\nabla V$, they can still be achieved by 2 queries of $P_{\nabla V}$. This is because we can apply two rounds of ``controlled'' rotation gadgets, each for a replica system, with the same gradient information in the ancilla sub-register (see Eq.~\cref{eqn:gradient-query-1}-\cref{eqn:gradient-query-3}). The block-encodings of $L_j$ and $L'_j$ have subnormalization factors
    \begin{equation}
        \alpha_1 = \pi N \sqrt{d}\beta^{-1/2}+\beta^{1/2}R/2,\quad \alpha_2 = \pi N \sqrt{d}\beta'^{-1/2}+\beta'^{1/2}R/2,
    \end{equation}
    respectively, where $R = \max_{x\in \Omega}\|V(x)\| \le \ell \sqrt{d}\log(d)$. Each set of operators $\{L_j\}^d_{j=1}$ and $\{L'_j\}^d_{j=1}$ requires $m_1 = 3$ ancilla qubits. According to~\cref{prop:block-encode-A}, the number of elementary gates is $\tilde{\mathcal{O}}(d^2)$.

    Now, we discuss how to block-encode the operator $L_s = \sqrt{\mu/2}(I-W)S^{1/2}$. Recall that $S$ is the multiplication operator given by $[S\psi](x,y) \coloneqq s(x,y)\psi(x,y)$.
    Note that the function $s(x,y)$ in~\cref{eqn:s-func-multiplicative} can be computed by $2$ queries to the function value of $V$ (one for $V(x)$, another for $V(y)$). 
    We can construct an arithmetic circuit that computes
    \begin{align}
        \ket{x,y}\ket{0,0,0} \mapsto \ket{x,y}\ket{V(x),V(y),\sqrt{s(x,y)}}.
    \end{align}
    Since $|s(x,y)|\le 1$, we can introduce an extra ancilla qubit and apply a sequence of rotations to generate an amplitude for the $\ket{0}$ state that equals the value of $\sqrt{s(x,y)}$:
    \begin{equation}
        \ket{x,y}\ket{V(x),V(y),\sqrt{s(x,y)}}\ket{0} \mapsto \ket{x,y}\ket{V(x),V(y),\sqrt{s(x,y)}} \left(\sqrt{s(x,y)}\ket{0} + \sqrt{1-s(x,y)}\ket{1}\right).
    \end{equation}
    Finally, we uncompute the middle register (i.e., $\ket{V(x), V(y), \sqrt{s(x,y)}}$) by inverting the arithmetic circuit computing $\sqrt{s(x,y)}$ and another 2 uses of the inverse $O_V$. This realizes a $(1,1)$-block-encoding of the multiplication operator $S^{1/2}$: 
    \begin{equation}
        \ket{x,y}\ket{0} \mapsto \sqrt{s(x,y)}\ket{x,y}\ket{0} + \ket{\perp}\ket{1},
    \end{equation}
    with a total number of $\poly(b)$
    elementary gates, where $b$ is the fixed-point precision to represent real numbers (e.g., $V(x)$, $\sqrt{s(x,y)}$).
    The swap operator $W$ can be implemented by $d\log_2(N)$ swap gates without ancilla qubits, where $N$ is the discretization number and $\log_2(N)$ is the number of qubits representing the quadrature states $\ket{x}$ (or $\ket{y}$). Using the standard Linear Combination of Unitaries (LCU)~\cite[Lemma 52]{gilyen2019quantum} , we can implement a $(1,1)$-block-encoding of $I-W$. Overall, we can implement a block-encoding of $L_s$ with $4$ queries to $O_V$, $\mathcal{O}(d\log(N))$ elementary gates, and $m_2 = 2$ ancilla qubits (one for the phase $\sqrt{s(x,y)}$, another for LCU). The subnormalization factor of this block-encoding is 
    \begin{equation}
        \alpha_3 = \sqrt{2\mu}.
    \end{equation}
    We can make all elementary gates in the block-encoding of $L_s$ to be controlled on $\ket{2d+1}$, then by concatenating this ``select-($2d+1$)'' oracle with a basis change that maps $\ket{0}$ to $\ket{2d+1}$, we obtain a unitary that maps $\ket{0^{m_3}}\ket{\psi}$ to $\ket{2d+1}(L_s/\alpha_3)\ket{\psi}+\ket{\perp}$, i.e., it block-encodes $L_s$ (with subnormalization factor $\alpha_3$) in the $(2d+1)$-th block in the first column.
    
    The last step is to construct the full block-encoding of $\mathbb{L}_{\mathrm{RE}}$ by stacking the above three parts together. We denote the the block-encodings for $\{L_j\}^d_{j=1}$, $\{L'_j\}^d_{j=1}$, and $L_s$ by $U_1$, $U_2$, and $U_s$, respectively.
    Now, we add $2$ ancilla qubits as the control register and form a ``select'' oracle of the form:
    \begin{equation}
        \UU = \ketbra{00}{00}\otimes U_1 + \ketbra{01}{01}\otimes U_2 + \ketbra{10}{10}\otimes U_s + \ketbra{11}{11}\otimes I.
    \end{equation}
    Now, we define a state $\ket{\bm{\alpha}} = \frac{\alpha_1}{\alpha}\ket{00}+\frac{\alpha_2}{\alpha}\ket{01}+\frac{\alpha_3}{\alpha}\ket{10}$ with a normalization factor $\alpha = \sqrt{\alpha^2_1+\alpha^2_2+\alpha^2_3}$, as given in~\cref{eqn:normalization-L-RE}.
    Note that $\ket{\bm{\alpha}}$ is a 2-qubit state and can be prepared by a circuit with $\mathcal{O}(1)$ elementary rotation gates. We denote this state preparation circuit as $\bm{P}$.  
    By applying the unitary operator $\UU$ on the state $\ket{\bm{\alpha}}\ket{0^{m}}\ket{\psi}$ with $m = 2m_1+m_2 = 8$, we end up with:
    \begin{equation}
        \UU \ket{\bm{\alpha}}\ket{0^{m}}\ket{\psi} = \ket{\bm{\alpha}}\ket{0^m}\left(\frac{1}{\alpha}\sum^{2d+1}_{j=1}A_j \ket{\psi}\right) + \ket{\perp},\quad A_j = \begin{cases}
            L_j, & j \in \{1,\dots,d\}\\
            L'_j, & j\in\{d+1,\dots, 2d\} \\
            L_s, & j \in \{2d+1\}.
        \end{cases}
    \end{equation}
    Finally, by uncomputing $\ket{\bm{\alpha}}$ by $\bm{P}^\dagger$, the circuit $\bm{P}^\dagger \UU \bm{P}$ implements a $(\alpha, 10)$-block-encoding of $\mathbb{L}_{\rm RE}$.
    In total, we have used $2$ queries to $O_{\nabla V}$, 4 queries to $O_V$, and an additional $\widetilde{\mathcal{O}}(d^2)$ elementary gates.
\end{proof}

\subsection{Proof of Theorem~\ref{thm:main-2}}\label{append:proof-main-2}

Similar to the previous section, we make the following assumption to ensure the efficiency of the spatial discretization of the generalized Witten Laplacian of RELD:

\begin{assump}[Spatial discretization of RELD]\label{assump:discretization-2}
    Let $V\colon \RR^d \to \RR$ be a smooth potential function, and $\bsigma \propto e^{-\beta V(x)-\beta'V(y)}$ be the joint Gibbs measure. 
    For an arbitrary $\epsilon>0$, we assume that we can choose $a = \mathcal{O}(\log(d/\epsilon))$ and $N = a\cdot\poly\log(d/\epsilon)$ such that the followings hold:
    \begin{enumerate}
        \item Let $\widetilde{\mathbb{L}}_{\rm RE}$ be discretized $\mathbb{L}_{\rm RE}$ (as in~\cref{eq:L_bb_RE}) described in~\cref{prop:block-encoding-L-RE}. The operator $\widetilde{\mathcal{H}}_{\rm RE} = \widetilde{\mathbb{L}}^\dagger_{\rm RE} \widetilde{\mathbb{L}}_{\rm RE}$ (i.e., discretized Witten Laplacian of RELD) has a ground state $\ket{\widetilde{\sqrt{\bsigma}}}$ that satisfies $\left\|I_N \ket{\widetilde{\sqrt{\bsigma}}} - \sqrt{\bsigma}\right\|_{L^2} \le \epsilon/4$, 
        \item Compared to the Witten Laplacian of RELD: $\mathcal{H}_{\rm RE} = \mathbb{L}^\dagger_{\rm RE}\mathbb{L}_{\rm RE}$, the smallest eigenvalue of $\widetilde{\mathcal{H}}_{\rm RE}$ is no greater than $\mathrm{Gap}\left(\mathcal{H}_{\rm RE}\right)/16$, and the second smallest eigenvalue of $\widetilde{\mathcal{H}}$ is no smaller than $9 \mathrm{Gap}\left(\mathcal{H}_{\rm RE}\right)/ 16$. 
    \end{enumerate}
\end{assump}

The following result is a detailed version of~\cref{thm:main-2}. 

\begin{thm}\label{thm:main-2-formal}
    Suppose that~\cref{assump:discretization-2} holds and the potential $V$ is $\ell$-smooth. 
    Assume access to a state $\ket{\bphi}$ (i.e., warm start) such that $\abs{\braket{\bphi}{\widetilde{\sqrt{\bsigma}}}} = \Omega(1)$.
    Let $\mathcal{L}$ be the generator of RELD as defined in~\cref{eqn:full-generator-reld}.
    Then, there exists a quantum algorithm that outputs a random variable $X\in \RR^d$ distributed according to $\eta$ such that $\mathrm{TV}(\eta, \sigma) \le \epsilon$ using 
    \begin{equation}\label{eqn:main-2-query-complexity}
        \sqrt{\frac{d}{\mathrm{Gap}\left(\mathcal{L}^\dagger\right)}} \cdot \left(1/\beta'+\ell^2\beta+\mu/d\right)^{1/2} \cdot \polylog(d, \epsilon^{-1})
    \end{equation}
    quantum queries to the function value $V$ and gradient $\nabla V$, respectively.
\end{thm}

\begin{proof}
    By~\cref{prop:block-encoding-L-RE}, we can construct a $(\alpha, 10)$-block-encoding of $\widetilde{\mathbb{L}}_{\rm RE}$ with a normalization factor 
    \begin{equation}
        \alpha \le 2\pi N \sqrt{d}\beta^{-1/2}+\beta^{1/2}R +\sqrt{2\mu}.
    \end{equation}
    Since $V$ is $\ell$-smooth and is restricted to the box-shaped numerical domain $\Omega$, we have $R = \max_{x\in \Omega} \|\nabla V(x)\| = \mathcal{O}(\ell \sqrt{d}\log(d))$. It turns out that
    \begin{equation}
        \alpha = \mathcal{O}\left(N\sqrt{d}\cdot \log(d) \cdot \left(1/\beta'+\ell^2\beta +\mu/d\right)^{1/2}\right).
    \end{equation}
    Recall that the right singular vector of $\widetilde{\mathbb{L}}_{\rm RE}$ associated with the smallest singular value is denoted by $\ket{\widetilde{\sqrt{\bsigma}}}$.
    Given access to a warm start state $\ket{\bphi}$ such that $\abs{\braket{\bphi}{\sqrt{\bsigma}}} = \Omega(1)$,
    by~\cref{thm:main-qsvt}, we can prepare a quantum state $\ket{\mathbf{g}}$ that is $\epsilon/4$-close to $\ket{\widetilde{\sqrt{\bsigma}}}$ with
    \begin{equation}\label{eqn:main-query-pre-re}
        \mathcal{O}\left(\frac{\alpha}{\sqrt{\mathrm{Gap}\left(\mathcal{L}^\dagger\right)}}\right) =  \sqrt{\frac{d}{\mathrm{Gap}\left(\mathcal{L}^\dagger\right)}} \cdot \left(1/\beta'+\ell^2\beta+\mu/d\right)^{1/2} \cdot \poly\log(\epsilon^{-1})\cdot \log(d)
    \end{equation}
    queries to the block-encoding of $\widetilde{\mathbb{L}}_{\rm RE}$, which amounts to the same query complexity to $O_V$ and $O_{\nabla V}$ due to~\cref{prop:block-encoding-L-RE}. By~\cref{assump:discretization-2} and the fact that $I_N$ is an isometry, we invoke the triangle inequality to show that
    \begin{equation}
        \|I_N \ket{\mathbf{g}} - \sqrt{\bsigma}\|_{L^2} \le \epsilon/2.
    \end{equation}
    Finally, by~\cref{lem:high-accuracy-interpolation}, we can realize a random variable $Z = (X,Y) \sim \bbeta$ 
    such that $\mathrm{TV}(\bbeta, \bsigma) \le \epsilon$ with the number of queries to $\nabla V$ and $V$ given by~\cref{eqn:main-2-query-complexity}. This immediately implies that the distribution of $X \sim \eta$, given by the marginal law of $\bbeta$, satisfies $\mathrm{TV}(\eta, \sigma) \le \epsilon$.
    This concludes the proof.
\end{proof}

\section{Details of Lindbladian-Based Warm-Start Preparation}\label{append:lindblad-warm-start}
In this section, we establish a direct connection between the Fokker--Planck equation~\cref{eqn:FKPK} and the Lindblad master equation~\cref{eqn:lindbladian}. In particular, by representing the density operator as a function in the form $\rho(t,x,y) \coloneqq \langle x |\rho(t)|y\rangle$, we show that the ``diagonal element'' of this function, namely, $\rho(t,x,x)$, solves the Fokker--Planck equation.

To prove this result, we first introduce some technical notations.
Let $w(x)\colon \RR^d \to \RR$ be a real-valued function, we denote $\hat{w}$ as the corresponding multiplicative operator acting on a test function $\varphi$ pointwisely, i.e., $(\hat{w}\varphi)(x) = w(x)\varphi(x)$.

\begin{lem}\label{lem:weak-convergence}
    Let $\rho(t)$ be the solution to the Lindblad master equation~\cref{eqn:lindbladian}, and $p(t,x)$ be the solution to~\cref{eqn:FKPK} with initial condition $p(0,x)$ given by the diagonal of $\rho(0)$. 
    Then, for any smooth function $w(x)\colon \RR^d \to \RR$ and $t \ge 0$, we have
    \begin{equation}\label{eqn:observable-equation}
    \Tr[\hat{w}\rho(t)] =   \int w(x) p(t, x)\d x.
\end{equation}
\end{lem}
\begin{proof}
    Since $p(0,x)$ matches the distribution given by $\rho(0)$, we have 
    \begin{align}
        \Tr[\hat{w}\rho(0)]=\int w(x)\rho(0,x,x)\d x = \int w(x)p(0,x)\d x,
    \end{align}
    which implies that~\cref{eqn:observable-equation} holds at $t=0$. Next, we note that
\begin{equation}
    \partial_t \Tr[\hat{w}\rho(t)] = \Tr\left[\hat{w} \mathfrak{L} [\rho(t)]\right] = \Tr\left[ \mathfrak{L} ^\dagger [\hat{w}]\rho(t)\right],
\end{equation}
where $\mathfrak{L}^\dagger$ denotes the adjoint of the Lindbladian operator $\mathfrak{L}$: for any observable $O$, we have
\[\mathfrak{L}^\dagger [O] = \sum^d_{j=1} \left(2 L^\dagger_j O L_j - \{L^\dagger_j L_j, O\}\right).\]
Recall that $ L_j =  -i\frac{1}{\sqrt{\beta}} \partial_{x_j} - i\frac{\sqrt{\beta}}{2}\partial_{x_j}V $,  $ L_j^\dagger  =  -i\frac{1}{\sqrt{\beta}} \partial_{x_j} + i\frac{\sqrt{\beta}}{2}\partial_{x_j}V $. 
For a fixed $j \in [d]$, direct calculation shows that
\begin{align}
    L^\dagger_j \hat{w} L_j - \frac{1}{2}\{L^\dagger_j L_j, \hat{w}\} 
    = 
    - \frac{1}{2}(\partial_{x_j} V)(\partial_{x_j} {w}) + \frac{1}{2\beta}(\partial_{x_j}^2 {w}).
\end{align}
Therefore, we have
\begin{equation}
    \mathfrak{L} ^\dagger [\hat{w}] = - \nabla V \cdot \nabla w + \frac{1}{\beta }\Delta w.
\end{equation}
It follows that 
\begin{align}\label{eqn:weak-sol}
    \partial_t \Tr[\hat{w}\rho(t)] &
    = \int \left(- \nabla V \cdot \nabla w + \frac{1}{\beta }\Delta w  \right)\rho(t,x,x)\d x \\
    &=  \int w(x)\left(\nabla \cdot ( p(t,x)\nabla V) + \frac{1}{\beta} \Delta p(t,x)\right)\d x,
\end{align}
where we write $p(t,x) = \rho(t,x,x)$ and the last step uses integration by parts. Comparing~\cref{eqn:weak-sol} with the Fokker--Planck equation~\cref{eqn:FKPK}, we find that $p(t,x)$ solves the Fokker--Planck equation and thus \cref{eqn:observable-equation} holds.
\end{proof}

\section{Details of Numerical Experiments} \label{append:detail_num}

\subsection{Numerical simulation of quantum algorithms} 
\label{append_meth_1}

In this subsection, we provide some details of the numerical simulation of our quantum algorithms on a classical computer. All experiments were conducted using MATLAB 2024b on a machine equipped with Intel(R) Core(TM) i7-14700K with 64 GB of memory.

\paragraph{Spatial discretization of $\mathbb{L}$.}
In our experiments, differential operators (e.g., $\partial_x$) are represented as pseudo-differential operators and computed using Discrete Fourier Transform (DFT).
With a regular mesh grid, operators defined as point-wise multiplication (e.g., $\partial_j V$ in LD or $S^{1/2}$ in RELD) are discretized and represented as diagonal matrices. For details, see the discussions in~\cref{append:spatial_discretize_Lj}.

\paragraph{Numerical implementation of singular value thresholding.}
The numerical implementation of the quantum singular value transformation (QSVT) involves approximating the desired operator transformations via polynomial approximations. Specifically, we construct a suitable Chebyshev polynomial summation $P(x)$ to approximate the relevant spectral filtering function by Fourier-Chebyshev expansion method~\cite[Section III.3]{dong2021efficient}:
To approximate a smooth real-valued function \( F(x) \) defined on the interval \([-1, 1]\), we express it in terms of Chebyshev polynomials of the first kind, i.e., 
\begin{align}
    F(x) \approx P(x) = \sum_{j=0}^{d} c_j T_j(x),\
    \label{eqn:FFT_cheby}
\end{align}
where \( T_j(x) \) denotes the Chebyshev polynomial of degree \( j \). The coefficients \( c_j \) can be efficiently computed using the fast Fourier transform (FFT) and a quadrature-based formula:
\begin{align}
    c_j \approx \frac{(2 - \delta_{j0})}{2K} (-1)^j \sum_{l=0}^{2K - 1} F(-\cos\theta_l) e^{ij\theta_l} 
\end{align}
Here, the quadrature nodes are given by \( \theta_l = \pi l / K \), where \( l \) ranges from \( 0 \) to \( 2K - 1 \), and \( K \) represents the total number of quadrature points. In our case, the spectral filtering function is a non-smooth rectangle function given by~\cref{eqn:rectangular-function}. Since the Fourier-Chebyshev expansion method applies to smooth functions, we first pre-process the rectangular filter function using cosine-based smoothing near transition points, yielding a smoothened filter function $F(x)$:
\begin{equation}
    \begin{aligned}
        F(x) = 
        \begin{cases}
        0, &x\in[-1,-s_1-\delta ] \cup [s_2+\delta ,1]\\
        \frac{1}{2} \left[ 1 - \cos\left( \pi \frac{x - (-s_1 - \delta)}{\delta} \right) \right], & x \in [-s_1 - \delta, -s_1], \\[6pt]
        1, & x \in [-s_1, s_2], \\[6pt]
        \frac{1}{2} \left[ 1 + \cos\left( \pi \frac{x - s_2}{\delta} \right) \right], & x \in [s_2, s_2 + \delta].
        \end{cases}
    \end{aligned}
\end{equation}
This new filter function $F(x)$ and its polynomial approximation $P(x)$ given by~\cref{eqn:FFT_cheby} are good approximations to the rectangular filter when $\delta $ is small. Meanwhile, the polymomial $P(x)$ satisfies~\cref{lem:rectangle-polynomial} after appropriately rescaling and shifting $P(x)$ to ensure $|P(x)|\leq 1$. 

To numerically implement singular value thresholding, we first compute the singular value decomposition of $\mathbb{L} = W\Sigma V^\dagger$ and apply the polynomial filtering function $P(x)$ to the singular values $\Sigma$. The desired (approximate) projector onto the encoded Gibbs state is given by $\tilde{\Pi} = V P(\Sigma) V^\dagger$, as described in~\cref{append:proof-main}. Finally, we apply the operator $\tilde{\Pi}$ to the warm start state $\ket{\phi}$ to obtain the desired state, which is a close approximation of the true Gibbs state $\ket{\sqrt{\sigma}}$.

\paragraph{Limitations of classical simulation.}
Despite these numerical strategies, classical simulations of the quantum algorithm remain fundamentally limited by the exponential growth of memory as dimension increases. In our experiments, the size of the component operators in the discretized $\mathbb{L}$ is $N^d$ for LD and $N^{2d}$ for RELD. In both cases, the matrix size grows exponentially with the problem dimension $d$. The numerical simulation becomes memory-intensive when involving (pseudo)differential operators of dimension $d \ge 3$. Thus, we implement numerical simulations of our quantum algorithms for two test problems: the 2-dimensional M\"uller-Brown potential (for LD) and a 1-dimensional non-convex potential (for RELD).

\subsection{Details of  quantum-accelerated Langevin dynamics} \label{append:detail_QALD}
The analytic expression of M\"uller-Brown potential~\cite[Footnote 7]{MuellerBrown1979} is given by
\begin{equation}
    V(x,y) = \sum_{k=1}^{4} A_k \exp\left[a_k(x - x_k^0)^2 + b_k(x - x_k^0)(y - y_k^0) + c_k(y - y_k^0)^2\right],
\end{equation}
with the constants:
\begin{align*}
A &= (-200, -100, -170, 15), & a &= (-1, -1, -6.5, 0.7), \\
b &= (0, 0, 11, 0.6), & c &= (-10, -10, -6.5, 0.7), \\
x^0 &= (1, 0, -0.5, -1), & y^0 &= (0, 0.5, 1.5, 1).
\end{align*}
The landscape of the potential is shown in~\cref{fig:muller-brown}. 

\begin{figure}[ht!]
    \centering
    \includegraphics[width=0.5\linewidth]{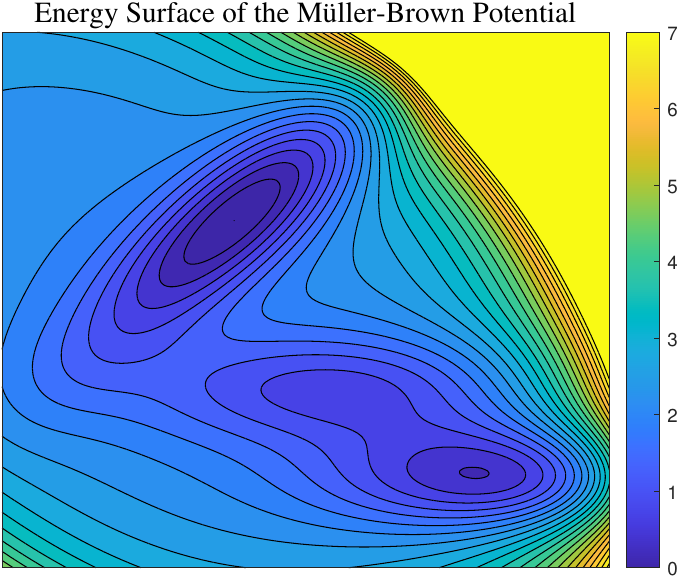}
    \caption{Energy surface of M\"uller-Brown potential  (with truncated color scale).}
    \label{fig:muller-brown}
\end{figure}

In numerical experiments of MALA, we select a time step size \(\Delta t = 10^{-3}\) to ensure numerical stability and adequate accuracy in sampling trajectories. The total number of generated samples is set to \(3\times10^4\), sufficiently large to approximate equilibrium distributions accurately. We explore various inverse temperature parameters $\beta \in \{0.4, 0.6, 0.8\}$. 
The MATLAB \texttt{histcounts2} function is utilized here to analyze the sampled data, which partitions the sampled values into $N=50$ discrete bins and counts the frequency of samples within each bin to get distribution. The frequency is used as an estimation of the sample distribution.

For quantum-accelerated Langevin dynamics, we set the number of grids to \(N = 50\) (or 50 Fourier modes) and try different degrees of polynomial to approximate the rectangle function~\cref{eqn:rectangular-function}. 
A subtle but important point is that we found $\mathrm{Gap}(\mathbb{L})/\alpha$ to be too small to be effectively filtered. This is because the normalization factor $\alpha$ becomes very large due to the influence of $R = \max_{\xx \in \Omega} \|\nabla V(\xx)\|$. As shown in~\cref{fig:muller-brown}, $R$ is primarily determined by the gradient in the top-right yellow region of the M\"uller-Brown potential. However, it is unlikely that MALA or quantum-accelerated Langevin dynamics would explore this region, as the potential values there are extremely high. To mitigate the effect of the normalization factor, we slightly modifiy the potential to $\max(V(x),5)$. This effectively reduces $\alpha$, which reduces the polynomial degree of the filter function without affecting the quality of the Gibbs state. Then, we follow the methodology described in~\cref{append_meth_1} to implement quantum-accelerated Langevin dynamics. The results are presented in~\cref{fig:qsvt-muller-brown}.

\subsection{Details of quantum-accelerated replica exchange} \label{append:detail_QARE}

In our numerical implementation of quantum-accelerated RELD, we set the number of grids (or Fourier modes) to be $N=150$.
We use two replicas: one corresponds to a low-temperature chain with inverse temperature \(\beta\), and another corresponds to a high-temperature chain with inverse temperature \(\beta'\). We fix the inverse temperature of the high-temperature chain \(\beta' = 1\) and vary \(\beta\) for the low-temperature chain from $2$ to $10$. The swapping intensity is also fixed as \(\mu = 1\). To compute the spectral gaps, we perform singular value decomposition on the operator \(\mathcal{H}\) (as in~\cref{eq:reld-witten-laplacian}) for classical RELD, and on \(\mathbb{L}_{\mathrm{RE}}\) (as in~\cref{eqn:L_RE}) for quantum-accelerated RELD.

\subsection{Details of Lindbladian-based warm-start preparation}

This part numerically solves \cref{eqn:lindbladian} when the potential function is \cref{eq:1D_hard_2850}. We adopt the spatial discretization method described in \cref{append_meth_1} with the number of grids (or Fourier modes) $N=50$. The time step size is $\Delta t=1\times 10^{-4}$, and the time propagation is conducted using the 4th-order Runge-Kutta method. The initial distribution is taken to be a very sharp Gaussian centered at $x=-1.7$, whose initial overlap with the Gibbs distribution is smaller than $3\times 10^{-2}$. Here, we set the inverse temperature $\beta$ from $2$ to $10$, which covers the range of the numerical experiment shown in~\cref{fig:re-spec-gap}. The overlap is calculated using~\cref{eqn:mixed-state-overlap}, and the numerical values are shown in~\cref{fig:warmVsMix}.

\end{document}